\documentclass[11pt]{article}
\usepackage[utf8]{inputenc}
\usepackage{amsthm}
\usepackage{amsmath}
\usepackage{amssymb}
\usepackage{nicefrac}
\usepackage{ulem}
\usepackage{thmtools} 
\usepackage{thm-restate}

\usepackage{fullpage}
\usepackage{tikz}
\usetikzlibrary{shapes}
\usetikzlibrary{positioning}
\usetikzlibrary{arrows}
\pgfarrowsdeclarecombine*{spaced stealth'}{spaced
stealth'}{stealth'}{stealth'}{space}{space}
\tikzset{>=spaced stealth'}

\newtheorem{theorem}{Theorem}
\newtheorem{lemma}[theorem]{Lemma}
\newtheorem{proposition}[theorem]{Proposition}

\newtheorem{definition}[theorem]{Definition}
\newtheorem{conjecture}[theorem]{Conjecture}

\newcommand{\OPT}{\mathrm{OPT}}
\renewcommand{\epsilon}{\varepsilon}
\newcommand{\eps}{\varepsilon}
\newcommand{\R}{\mathbb{R}}
\newcommand{\E}{\mathbb{E}}
\newif\ifnotes\notestrue

\makeatletter
\renewcommand\paragraph{\@startsection{paragraph}{4}{\z@}%
                                    {1.5ex \@plus1ex \@minus.2ex}%
                                    {-1em}%
                                    {\normalfont\normalsize\bfseries}}

\makeatother

\ifnotes
 \usepackage{color}
 \definecolor{mygrey}{gray}{0.50}
 \newcommand{\notename}[2]{{\textcolor{red}{\footnotesize{\bf (#1:} {#2}{\bf ) }}}}

\else
 
 \newcommand{\notename}[2]{{}}
 
\fi

\title{Flow Time Scheduling and Prefix Beck-Fiala}
\author{Nikhil Bansal \and Lars Rohwedder \and Ola Svensson}
\date{}

\begin{document}

\maketitle

\begin{abstract}

We relate discrepancy theory with the classic scheduling problems of minimizing  max flow time and total flow time on unrelated machines. Specifically, we give a general reduction that allows us to transfer discrepancy bounds in the prefix Beck-Fiala (bounded $\ell_1$-norm) setting to 
bounds on the flow time of an  optimal schedule.

Combining our reduction with a deep result proved by Banaszczyk via convex geometry, give  guarantees of $O(\sqrt{\log n})$
 and $O(\sqrt{\log n} \log P)$ for max flow time and total flow time, respectively,   improving upon the previous best guarantees of $O(\log n)$ and $O(\log n \log P)$.
 Apart from the improved guarantees, the reduction motivates seemingly easy versions of prefix discrepancy questions: {any} constant bound on prefix Beck-Fiala where vectors have sparsity two (sparsity one being trivial)  would already yield tight guarantees for both max flow time and total flow time. While known techniques solve this case when the entries take values in $\{-1,0,1\}$, we show that they are unlikely to transfer to the more general $2$-sparse case of bounded $\ell_1$-norm.

\end{abstract}

\thispagestyle{empty}
\clearpage
\newpage
\setcounter{page}{1}

\section{Introduction}

In this paper we formally relate flow time scheduling with prefix discrepancy, yielding new improved bounds on classic  scheduling problems and new directions in discrepancy. 

The scheduling problems that we consider are in the most general \emph{unrelated machine model}: The input consists of  a set $J$  of $n$ jobs, a set $M$ of $m$ machines, processing times $(p_{ij})_{i\in M, j\in J}$, and release times $(r_j)_{j\in J}$. The machines are unrelated in the sense that the processing time  $p_{ij} \geq 0$ of a job $j\in J$ can arbitrarily depend on the machine $i\in M$. 
 A schedule processes each job $j$  on a selected machine $i$ 
for a total time $p_{ij}$ after its release time $r_j$. A machine can process at most one job at any time but we do allow for preemption, i.e.,  a job can be preempted and resumed at a later stage (on the same machine)\footnote{This is a necessary assumption to get any meaningful guarantees for the total flow time objective. For maximum flow time, one can always transform a preemptive non-migratory schedule into a non-preemptive one
(that is, non-preemptiveness comes for free).}. 
When the schedule is clear from the context, we denote by
$F_j$ the flow time of a job, which is the duration from its release time until it is completed.
Two central and well-studied objectives are to minimize maximum flow time, $\max_{j\in J} F_j$, and to minimize total flow time,  $\sum_{j\in J} F_j$.

Flow time objectives are notoriously difficult and it remains a major open question to understand the approximability  of unrelated machine scheduling under both the max flow time and total flow time objectives. 
The best upper bounds are by Bansal and Kulkarni~\cite{DBLP:conf/stoc/BansalK15}, who obtained the approximation guarantees $O(\log n)$ and $O(\log n \log P)$ for max flow time and total flow time, respectively. 
Here $P$ denotes the ratio between the largest and smallest (finite) processing time, which can be upper bounded by $\text{poly}(n)$ using standard arguments.  The approach of~\cite{DBLP:conf/stoc/BansalK15} is based on an iterative rounding scheme of the natural linear programming (LP) relaxations and, as we further elaborate on below, these techniques are unlikely to lead to better guarantees. At the same time, the best known hardness of approximation results say that it is \textsf{NP}-hard to approximate the max flow time objective  better than a factor $3/2$~\cite{DBLP:journals/mp/LenstraST90} and the total flow time objective within a factor better than $O(\log P/\log \log P)$~\cite{DBLP:conf/icalp/GargK06,DBLP:conf/isaac/GargKM08}. 
In summary,  the best known guarantees are roughly a factor $\log n$ away from the known hardness results.

We now explain the difficulties and limitations of current techniques by focusing on the max flow time objective but much of what is said also applies to total flow time.  
Maximum flow time generalizes another classic scheduling problem: makespan minimization where every job is released at time $0$ and one wants to minimize the maximum completion time.  
In a seminal work, Lenstra, Shmoys, and Tardos~\cite{DBLP:journals/mp/LenstraST90} gave
a beautiful $2$-approximation algorithm for this special case (and proved the lower bound of $3/2$ which also remains the best lower bound for max flow time). Their algorithm is based on the insight that any extreme point solution to the natural LP relaxation has  few fractionally assigned jobs, 
as the relaxation has relatively few constraints. 
They then show that these fractional jobs can be distributed among the machines so that each machine receives at most one additional job, leading to the approximation guarantee of $2$. 

\paragraph{An illustrative example for max flow.}
However, the generalization to max flow becomes much harder.
The following example is illustrative.
Consider a makespan instance $I$ 
where the optimal solution has makespan $T$ and each machine also has load $T$.
Now consider a max flow time instance where a copy of $I$ is released
at times $T,2T,3T,\dotsc,tT$. Note that the optimal max flow is still $T$. However,
in order to find a $c$-approximation, not only should the solutions to each sub-instance $I$ be $c$-approximations, but we also have to
ensure that the error in these instances does not accumulate over time. Otherwise, the jobs released late will be delayed and incur immense flow times.

In fact, for a $c$-approximation we need that
in \emph{every interval} of time on each machine the total error is only $(c-1)T$. This requires the LP relaxation for max flow time to have a constraint for every time interval, which increases the number of constraints by a polynomial factor compared to the relaxation for makespan. Extreme points are therefore less  
sparse  and each step of the rounding is only able to integrally assign half the jobs. This  naturally leads to an iterative rounding procedure that is repeated $O(\log n)$ times (to assign all jobs) and the increased approximation guarantee.

\paragraph{Prefix Discrepancy.}
There is a close connection between discrepancy theory and the problem of rounding fractional solutions \cite{DBLP:journals/ejc/LovaszSV86}.
Not surprisingly, similar difficulties to those stated above arise in discrepancy problems when bounding the discrepancy over intervals, due the accumulation of error.

In the typical discrepancy setting, we are given a collection of vectors 
$v^{(1)}, v^{(2)} , \ldots, v^{(n)} \in \mathbb{R}^m$ and the goal is to find signs $\epsilon_1,\epsilon_2,\dotsc,\epsilon_n\in\{-1, 1\}$ such that the $\ell_\infty$-norm of the signed sum  
$\epsilon_1 v^{(1)} + \epsilon_2 v^{(2)} + \cdots + \epsilon_n v^{(n)}$
is as small as possible.
In a seminal  work,
 Beck and Fiala \cite{BeckFiala-DAM81} showed the following general result\footnote{Typically this result is stated for vectors satisfying both $\|v^{(i)}\|_1 \leq t$ and $\|v^{(i)}\|_\infty \leq 1$. 
 Beck and Fiala proved a guarantee of $2t$ in this setting and a major open problem is to improve this bound to $O(\sqrt{t})$. However, we will only focus on vectors with bounded  $\ell_1$-norms.}.
 Given an arbitrary collection of vectors $v^{(1)}, v^{(2)} , \ldots, v^{(n)} \in \mathbb{R}^m$ of bounded $\ell_1$-norm $\| v^{(j)} \|_1\leq 1$, there always exist signs  $\epsilon_1,\epsilon_2,\dotsc,\epsilon_n\in\{-1, 1\}$ such that 
\begin{equation*}
    \lVert \epsilon_1 v^{(1)} + \epsilon_2 v^{(2)} + \cdots + \epsilon_n v^{(n)} \rVert_\infty \le C \ ,
\end{equation*}
where $C$ can actually chosen to be $2$, i.e.,~independent of the number $n$ of vectors and the dimension $m$. The proof of this result is also based on iterated rounding and in fact very closely related to the proof in \cite{DBLP:journals/mp/LenstraST90} for makespan scheduling.


Let us pursue this connection further, 
and consider the setting  where we want low discrepancy $\|\sum_{i\in I} \epsilon_i v^{(i)}\|_\infty $  for any consecutive subset $I\subseteq \{1, \ldots, n\}$ of indices. This is equivalent up to a factor $2$ to the \emph{prefix Beck-Fiala} problem where given any set of vectors $v^{(1)}, v^{(2)}, \ldots, v^{(n)} \in \mathbb R^m$ of bounded $\ell_1$-norm $\|v^{(j)}\|_1 \leq 1$, we wish to find  signs $\epsilon_1,\epsilon_2\dotsc,\epsilon_n\in\{-1, 1\}$ satisfying for every prefix $k=1, 2, \ldots, n$ that
\begin{equation*}
\lVert \epsilon_1 v^{(1)} + \epsilon_2 v^{(2)} + \cdots + \epsilon_k v^{(k)} \rVert_\infty \le C \ .
\end{equation*}
Interestingly, one runs into the same problem in adapting the proof technique of Beck and Fiala to this prefix version, as in going from makespan to max flow time scheduling.
That is, the linear algebraic techniques give a $C= O(\log n)$ bound on the prefix Beck-Fiala discrepancy. 

However, in contrast to flow time scheduling, more powerful techniques are known to yield better bounds on the prefix Beck-Fiala discrepancy. 
Specifically, Banaszczyk~\cite{DBLP:journals/rsa/Banaszczyk12} developed an ingenious
technique using deep ideas from convex geometry, that allowed him, among other things, to bound the prefix Beck-Fiala discrepancy by $O(\sqrt{\log n})$.
Interestingly, we show that these techniques can be transferred to flow time scheduling to obtain interesting new results. Conversely, this connection also leads to interesting new questions in discrepancy theory.

\subsection{Results}
Our main result is a general reduction that  allows us to transfer the  techniques from discrepancy to flow time scheduling.
\begin{theorem}
\label{thm:main}
    If the discrepancy of the prefix Beck-Fiala problem is bounded by $C$, then integrality gaps of the standard LP relaxations are upper bounded by $O(C)$ and $O(\min\{\log n, \log P\} \cdot C)$ for max flow time and total flow time, respectively.
\end{theorem}
Using $C= O(\sqrt{\log n})$ by the result of Banaszczyk~\cite{DBLP:journals/rsa/Banaszczyk12}, this gives improved bounds on the integrality gaps of  $O(\sqrt{\log n})$ for maximum flow time and a $O(\min\{\log n, \log P\} \cdot \sqrt{\log n})$ for total flow time.
The prefix Beck-Fiala problem and its further generalization called the prefix Koml\'os problem (discussed later), are interesting problems on their own with several other applications. It is been conjectured that the discrepancy for these problems and other related problems on prefix discrepancy may be $O(1)$ \cite{Spencer86Prefix, DBLP:journals/rsa/Banaszczyk12, bjmss22-itcs}.
If this conjecture is true, then Theorem \ref{thm:main} would imply tight integrality gaps of $O(1)$ and $O(\min\{\log n, \log P\})$ for maximum and total flow time, respectively.

The idea in the proof is to define a prefix Beck-Fiala instance based on the fractional solution with one vector per job, 
and given a low discrepancy $\pm 1$ coloring of vectors, use the signs to determine which machine to assign the  corresponding job to. 
Of course, in general a job might have $m$ potential machines where it can be assigned, which is not a binary decision.
The first part in both proofs is to reduce the problem of rounding a general to that of rounding a half-integral solution.
In a half-integral solution each job has only two choices and the rounding problem can be related to 
discrepancy in a clean way. This reduction from general to half-integral solutions is quite standard in discrepancy, see e.g.~\cite{DBLP:journals/ejc/LovaszSV86}, however our reduction is somewhat different and requires more care.
Additional difficulty arises in total flow time, because the linear program does not naturally
give rise to an order on the jobs, for example,
by release time, in which prefix discrepancy
should be applied.
We address this by preprocessing the LP solution so that there is one consistent order across all machines. Although we cannot enforce this order between any two jobs in the LP schedule, we can for jobs of similar sizes within a machine, which in turn suffices for our proof.

\paragraph{Algorithmic aspects.}

Our reduction in Theorem \ref{thm:main} is constructive in both cases: if a discrepancy $C$ coloring for prefix Beck-Fiala can be constructed in polynomial time, we get a $O(C)$-approximation algorithm for max flow time and a $O(\min\{\log(n), \log(P)\} \cdot C)$-approximation algorithm for total flow time.

However, Banaszczyk's proof does not imply an efficient algorithm
that recovers the signs $\epsilon_1,\dotsc,\epsilon_n$. 
While there has been a lot
of progress on making various techniques in discrepancy theory constructive~\cite{DBLP:conf/focs/Bansal10, DBLP:journals/siamcomp/LovettM15, DBLP:conf/stoc/BansalG17, DBLP:journals/siamcomp/BansalDG19, DBLP:journals/toc/BansalDGL19}, the case of prefix discrepancy remains elusive. The best known bound for prefix Beck-Fiala that is achievable in polynomial time is $O(\log n)$, and making progress here is an interesting open problem.
Thus our bounds in Theorem \ref{thm:main} do not give better constructive approximation guarantees over those already known.
However, they do give improved efficient
estimation algorithms, that is, algorithms which approximate the value of the optimal
solution up to a multiplicative error. The optimum of the LP relaxation, which can
be computed in polynomial time, gives such an estimation.


\paragraph{Conjectures for prefix discrepancy.}
The result of Banasczcyk for prefix discrepancy actually applies (with a bound of $O(\sqrt{\log nm})$ to the more general setting where we only assume that the vectors $v^{(1)}, v^{(2)}, \ldots, v^{(n)}$ have bounded $\ell_2$-norm $\|v^{(j)}\|_2\leq 1$ (note that $\|v\|_2 \leq \|v\|_1$ for any vector $v$). With this weaker assumption, already a major open problem, known as Koml\'os conjecture, is whether the discrepancy is constant in the setting without prefixes. 
Yet, it is plausible that the discrepancy is constant even for the prefix version of the Koml\'os problem.
As some supporting evidence for this conjecture, we consider a natural SDP relaxation of the prefix Koml\'os problem, and show that the SDP-discrepancy is bounded by $1$.
\begin{theorem}
  The SDP-discrepancy of the prefix Koml\'os problem is at most $1$.
\end{theorem}
This SDP discrepancy bound could also be of interest as it could potentially give a way to find a better
constructive bound for the prefix Koml\'os or prefix Beck-Fiala problem.
While it seems extremely ambitious to further improve  Banaszczyk's $O(\sqrt{\log nm})$ bound for prefix Koml\'os (given the current status of the Koml\'os conjecture), 
improving the bounds for prefix Beck-Fiala may be easier
and thus we want to emphasize this open question here.
\begin{conjecture}
  The discrepancy for the prefix Beck-Fiala problem is bounded by a constant.
  \label{conj:prefix_beck_fiala}
\end{conjecture}

In fact, in our reductions in Theorem~\ref{thm:main}, the vectors in resulting prefix Beck-Fiala instances have sparsity only two (i.e.,~two non-zero entries). So we highlight this seemingly very special case of Conjecture \ref{conj:prefix_beck_fiala}. 
\begin{conjecture}
\label{conj:2-sparse}
  The discrepancy for the prefix Beck-Fiala problem where each vector $v^{(i)}$ has sparsity $2$ is bounded by a constant.
  \label{conj:prefix_two_sparse}
\end{conjecture}
Proving Conjecture \ref{conj:2-sparse} together with Theorem~\ref{thm:main} would give tight bounds on the integrality
gap in both variants of flow time (up to constants), and furthermore an algorithmic proof of this conjecture would give an optimal algorithm for the flow time problems (unless $\mathsf{P}=\mathsf{NP}$); a weaker but still very interesting question would be to make Banasczcyk's arguments constructive in this special case which would then lead to a $O(\sqrt{\log n})$-approximation algorithm for max flow time and a  $O(\min\{\log(n), \log(P)\} \cdot \sqrt{\log n})$-approximation algorithm for total flow time. Moreover, an affirmative solution to Conjecture~\ref{conj:2-sparse}
even for the case where each vector has only two non-zero entries, which are negations of each other ($a$ and $-a$),
would already settle the restricted assignment variant ($p_{ij} \in \{p_j, \infty\}$), which is open for
max flow time.

\paragraph{From max flow to total flow.}

Our reduction from flow time scheduling to prefix Beck-Fiala actually shows an equivalence between max flow time scheduling and a very special case of 2-sparse prefix Beck-Fiala. 
We make this equivalence explicit and use it to relate the two flow time objectives:
\begin{theorem}\label{thm:relation}
    If the integrality gap of the standard LP of max flow time is upper bounded by $C$, then the integrality gap of the standard LP of total flow time is upper bounded by $O(\min\{\log(n), \log(P)\} \cdot C)$.
\end{theorem}
In particular, improving our bound for max flow time would immediately imply also an improvement for total flow time.

Interestingly, the 2-sparse case of the prefix Beck-Fiala problem has been studied before in the further special case where the vectors have entries $\{0,1\}$, due to the close connection with the classical $2$-permutation problem in discrepancy, and it is known that the prefix discrepancy here is at most $1$. Moreover, this can be achieved by a simple algorithm.
These techniques further generalize to the case when the values are in $\{-1,0,1\}$ to give a prefix discrepancy of $O(1)$ \cite{sahil-personal}.
This makes Conjecture  \ref{conj:2-sparse} even more plausible.

\paragraph{Maker-Breaker games.} Perhaps surprisingly, all the known techniques for $2$-sparse vectors with $\{-1,0,1\}$ entries seem to break down completely for general values in $[-1,1]$.
We explore this further and show that those techniques are unlikely to extend to general $2$-sparse vectors of bounded $\ell_1$-norm.
Specifically,
we show that the techniques extend
naturally to a related maker-breaker discrepancy game. A good strategy for this game implies a small prefix discrepancy for $2$-sparse vectors.
Yet, we formally separate the classes of $\{-1,0,1\}$ entries and $[-1,1]$ entries
in this game, showing an impossibility result for the latter.
\begin{theorem}
 In the maker-breaker game (defined in Section~\ref{sec:game}) the maker can maintain constant discrepancy if all entries take integer values in $\{-1, 0, 1\}$, whereas no strategy can obtain a discrepancy $o(\log n/\log\log n)$ in the presence of fractional values in $[-1, 1]$.
\end{theorem}
Thus, a proof for low prefix discrepancy must make use of information not available in this game, which the known techniques in the $\{-1,0,1\}$ case do not.

\subsection{Outline}

Starting with our main results, we prove
in Sections~\ref{sec:max_flow_time} and~\ref{sec:total_flow_time} the bounds
on the integrality gaps for max flow time and total flow time.
In Section~\ref{sec:relation} we then take a closer look at the precise
discrepancy bound necessary for the aforementioned proofs to hold.
We point out that a special case of $2$-sparse vectors is essentially equivalent
to the integrality gap for max flow time and sufficient for the proof of total
flow time. Consequently, any improvement on max flow time would yield also
an improvement for total flow time.
Then in Section~\ref{sec:game} we study a discrepancy game that
provides evidence that the two sparse case does
not easily follow from a clean approach that solves the case when additionally
all entries are in $\{-1, 0, 1\}$.
Finally, in Section~\ref{sec:SDP} we present a modified version of a proof due
to Raghu Meka, which, by applying Banasczcyk's technique on a non-trivial
convex body, shows that the vector relaxation of prefix Beck-Fiala
and the stronger prefix Koml\'os has a constant discrepancy.

\subsection{Related work}
The literature on both flow time and discrepancy is extensive and we only mention the most relevant results
on offline approximations for flow time and on prefix discrepancy.

\paragraph*{Total flow time.}
Already in simple case of multiple machines,
total flow time is hard to approximate better than $O(\log(P) / \log\log(P))$, see~\cite{DBLP:conf/focs/GargK07, DBLP:conf/isaac/GargKM08}.
On the other hand $O(\log(P))$-approximations are known in various settings.
First, this was shown for identical machines ($p_{ij} = p_j$)~\cite{DBLP:journals/jcss/LeonardiR07, DBLP:conf/stoc/AwerbuchALR99}.
For related machines ($p_{ij} = p_j / s_i$) the same guarantee was obtained in~\cite{DBLP:conf/icalp/GargK06}.
An $O(\log(P))$-approximation is also known for the restricted assignment setting where processing times satisfy $p_{ij}\in\{p_j,\infty\}$~\cite{DBLP:conf/focs/GargK07}. 
This result follows from an extension of
the single-source unsplittable flow problem \cite{DinitzGG99}. However, this method depends on the processing times
of a job being the same on all eligible machines.
As mentioned before, in the most general unrelated machine model there is
an $O(\log n \log P)$-approximation~\cite{DBLP:conf/stoc/BansalK15},
which we (non-constructively) improve upon.

Another line of work that has received significant attention recently is concerned
with minimizing the sum of weighted flow times. Since the problem becomes
very hard already on two machines, the main open problem was to
get a constant approximation in the single machine setting,
which has recently been achieved with the currently best rate being~$2+\epsilon$~\cite{DBLP:journals/siamcomp/BansalP14, DBLP:conf/focs/Batra0K18, DBLP:conf/soda/FeigeKL19, DBLP:conf/stoc/RohwedderW21}.

\paragraph*{Max flow time.}
For identical machines a greedy algorithm yields a $3$-approximation~\cite{DBLP:conf/soda/BenderCM98, DBLP:journals/ijfcs/Mastrolilli04}.
It is an intriguing question whether this can be improved on; to the best of our knowledge even a PTAS could exist.
A constant approximation is also known for related machines~\cite{DBLP:journals/toc/BansalC16}.
For unrelated machines is the state-of-the-art is an
$O(\log n)$-approximation algorithm~\cite{DBLP:conf/stoc/BansalK15},
for which we give a (non-constructive) improvement.
Note that even in the restricted assignment setting, no other results are known.
The unsplittable flow approach used in total flow time does not seem to help in this case, which can
perhaps be explained by the fact that the proof uses amortization between machines, which is unsuitable
for the hard constraints in max flow time.
Morell and Skutella conjecture that a strengthening of the
result for unsplittable flow is possible that indeed would be
sufficient for the restricted assignment case~\cite{MorellS21}.
However, the current status of the conjecture does not imply any meaningful guarantee.

\paragraph{Prefix discrepancy.}
Prefix discrepancy is also widely known as the \emph{signed series}
problem.
The prefix Koml\'os problem was introduced by Spencer~\cite{Spencer77} who showed that
there always exists a coloring with a prefix discrepancy bound that only depends on the dimension $m$ (i.e., independent of the number of vectors $n$).
This was later improved to $O(m)$ by Barany and Grinberg~\cite{BaranyGrinberg81}.
In fact, their result is more general in the way that it gives
a bound of $O(m)$ for the norm of any prefix sum
if the input vectors are bounded in the same norm. This holds for
any norm.

Since then prefix discrepancy has been studied a lot and has found several important applications. 
For instance, it implies a bound on the classical Steinitz problem~\cite{Chobanyan94,Steinitz-16} on the rearrangement of vector sequences, it arises naturally in online discrepancy problems where one is interested in bounding the discrepancy at all times~\cite{BJSS20,BJMSS-SODA21,ALS-STOC21}, and it also implies the best known bound for Tusn\'ady's problem~\cite{Nikolov-Mathematika19}.
Banaszczyk~\cite{DBLP:journals/rsa/Banaszczyk12} showed that prefix Koml\'os admits a coloring of $O(\sqrt{\log mn})$ discrepancy, thereby exponentially improving the dependency on $m$ in the Barany-Grinberg bound, but incurring a $\sqrt{\log n}$ dependence on the number of vectors $n$. The bound simplifies to $O(\sqrt{\log n})$ for prefix Beck-Fiala.


\paragraph*{Other applications of discrepancy in scheduling and packing.} 
Discrepancy was used in a breakthrough result~\cite{DBLP:conf/focs/Rothvoss13, DBLP:conf/soda/HobergR17}
to give an additive $O(\log n)$-approximation for the classical bin packing problem.
It was also used to get an improved guarantee of~$O(\log^{\nicefrac 3 2} n \cdot \mathrm{polyloglog}\,n)$
for the broadcast scheduling problem~\cite{DBLP:conf/soda/BansalCKL14}.
These results where based on a different partial coloring approach based on the entropy method. In our flow time applications in this paper, these techniques seem to inherently lose a $O(\log n)$ factor instead of $O(\sqrt{\log n})$ and do not improve existing results. 

The Steinitz problem, which is closely related to prefix discrepancy,
has many seen applications, including flow shop and job shop scheduling \cite{Barany81, Sevastjanov94} and faster algorithms for solving integer programs~\cite{BMMP12,EW18,JR-ITCS19}
and these integer programming algorithms themselves are applied in some scheduling problems, see for example~\cite{BerndtDJRAlenex}.

\section{Max flow time}
\label{sec:max_flow_time}
In this section we show that the natural assignment LP for max flow time has integrality gap $O(\sqrt{\log n})$.

\subsection{The assignment LP}
\label{sec:LP}
We formulate a natural assignment LP that was also used in~\cite{DBLP:conf/stoc/BansalK15}, and is
a generalization  of the well known linear program
for makespan minimization ($r_j = 0$ for all $j\in J$),
due to Lenstra, Shmoys, and Tardos
~\cite{DBLP:journals/mp/LenstraST90}.

Before that, let us start with some definitions.
Let $T$ be a parameter which specifies a fixed bound on
the flow time of all jobs.
The optimum LP solution is the smallest $T$ for which the LP below is feasible.
Such a $T$ can be found by a binary search over the range $[1,np_{\max}]$.
Here $p_{\max}$ denotes the maximum over all finite $p_{ij}$.
Note that we allow $p_{ij} = \infty > p_{\max}$ to
model that job $j$ cannot be scheduled on machine $i$.
We may also assume that $p_{\max} \le T$, as if $p_{ij} > T$ then no feasible solution
can assign job $j$ to machine $i$ and we can set $p_{ij}=\infty$.

For each $j\in J$ and $i\in M$, let the variable $x_{ij}$ indicate  whether job $j$ is assigned
to machine $i$.
The assignment LP is obtained by relaxing $x_{ij}$ to be fraction as is defined as follows.
\begin{align*}
    \sum_{i \in M} x_{ij} &= 1 &\forall j\in J \\
    \sum_{j \in J : r_j \in [t_1, t_2]} x_{ij} p_{ij} &\le t_2 - t_1 + T &\forall t_1 \le t_2, \quad \forall i\in M \\
    x_{ij} &= 0 &\forall i\in M, j\in J \text{ with } p_{ij} > T \\
    x_{ij} &\ge 0
\end{align*}
The first set of constraints enforce that
each job is assigned to exactly one machine.
The second set of constraints bound the volume of jobs
released during some interval $[t_1,t_2]$ of time and processed on the same machine. We will refer to these as interval constraints.
These constraints
are clearly necessary for any integral solution with flow time at most $T$ as
all jobs released during $[t_1, t_2]$ must be completed before $t_2 + T$,
or else the flow time of some job would exceed $T$.
Thus, the total volume of such
jobs on each machine (the left hand side of the constraint) can be at most
the total volume the machine can process during $[t_1, t_2 + T]$
(the right hand side). Notice that it suffices to have the constraints only for times $t_1,t_2$ when some job is released, so there are only $n^2$ such choices.

On the other hand, these constraints are also sufficient (for an integral
solution) to guarantee that there is a schedule with max flow
time $T$; 
simply considering the jobs in the order of their release times
and scheduling each job at the earliest possible time.


\subsection{Reducing to half-integral solutions}
We now reduce the task of rounding general LP solutions
to that of half-integral LP solutions.

\begin{lemma}
\label{lem:maxflow-reduction}
  Suppose for any instance the additive integrality gap for half-integral solutions
  to the assignment LP is bounded by $f(n, m) \cdot p_{\max}$ for some non-decreasing function $f$.
  Then the additive integrality gap of the assignment LP is at most $O(f(n^2, m) \cdot p_{\max})$.
  In particular, its (multiplicative) integrality gap is $O(f(n^2, m))$.
\end{lemma}
\begin{proof}
  Let $x^*$ be an optimal solution to the assignment LP of value $T^*$.
  We can construct another solution $x^{(\ell)}$ of value $T^{(\ell)} = T^* + p_{\max}$ where all variables are integer multiples of $1/2^\ell$ for
  $\ell = \lceil \log_2 n \rceil$:
  We start with $x^{(\ell)} = x^*$ and iteratively select a job $j$ that has variables that
  are not multiples of $1/2^\ell$.
  Note that since $\sum_{i\in M} x^{(\ell)}_{ij} = 1$, there must be at least two such variables for $j$. We increase one and decrease the other
  by the smallest margin so that one of them becomes a multiple of $1/2^\ell$. This is repeated until all variables are multiples of $1/2^\ell$.
  Notice that for any $t_1 \le t_2$ and $i\in M$, the value of
  $\sum_{j\in J : r_j\in [t_1, t_2]} x^{(k)}_{ij} p_{ij}$ can increase
  by at most $p_{\max} \cdot n / 2^\ell \le p_{\max}$.
  
  Next, we show inductively that there is a solution $x^{(h)}$ for all $h = \ell, \ell-1, \dotsc, 0$,
  where all the variables are integer multiples of $1/ 2^h$ and the
  objective is $T^{(h)} \le T^* + f(n^2, m) \cdot p_{\max} / 2^{h-1} + p_{\max}$.
  By construction, the solution $x^{(\ell)}$
  satisfies the base case. 
  
  Let us now assume we are given such a solution $x^{(h)}$
  for some $h$ and prove the existence of the solution $x^{(h-1)}$.
  To do this we create a new scheduling instance $I'$ and use the integrality gap guarantee for half-integral solutions for $I'$.
  To this end, for each job $j$ we form $2^{h-1}$ pairs $\{i_1, i_2\}$ such that each machine appears in
  exactly $2^h \cdot x^{(h)}_{ij}$ pairs.
  We construct now a new instance $I'$ with one job $j' = j'(j,i_1,i_2)$
  for each job $j$ and each such pair $\{i_1, i_2\}$.
  So the instance $I'$ has $2^{h-1}n \leq n^2$ jobs. 
  This $j'$ will be scheduled only on $i_1$ and $i_2$ and its processing times are defined by $p'_{i_1 j'} = p_{i_1 j} / 2^{h-1}$,
  $p'_{i_2 j'} = p_{i_2 j} / 2^{h-1}$, and $p_{i j'} = \infty$ for all $i \notin \{i_1, i_2\}$.
  So the new maximum finite processing time is $p'_{\max} \le p_{\max} / 2^{h-1}$.
  Moreover, the solution $x'_{i_1 j'} = x'_{i_2 j'} = 1/2$ is feasible for the LP with value $T^{(h)}$.
  An integral solution for $I'$ will now correspond to a solution of the original instance where each variable is a multiple of $1/2^{h-1}$.
  By the additive integrality gap for half-integral solutions
  and the induction hypothesis, 
  there is such a solution with value
  \begin{align*}
      T^{(h-1)} & \le T^{(h)} + f(n^2, m) \cdot p'_{\max} \\
      &\le T^* + f(n^2, m) \cdot p_{\max} / 2^{h-1} + p_{\max} + f(n^2, m) \cdot p_{\max} / 2^{h-1} \\
      &= T^* + f(n^2, m) \cdot p_{\max} / 2^{h-2} + p_{\max} \ .
  \end{align*}
  Thus, $x^{(0)}$ forms an integral solution with value $T^* + (2 f(n^2, m) + 1) p_{\max}$.
\end{proof}

\subsection{Rounding half-integral solutions}
We now bound the integrality gap for half-integral solutions.
\begin{lemma}
\label{lem:maxflow-halfintegral}
  Let $g(n, m)$ be a non-decreasing bound on the discrepancy of the prefix Beck-Fiala problem. Then, for any half-integral solution $x^*$ to the assignment LP with value $T^*$, 
  there exists an integral solution of value $T^* + O(g(n, m) \cdot p_{\max})$.
\end{lemma}
\begin{proof}
Let $J_2$ be the jobs that are not integrally assigned in $x^*$.
For each job $j \in J_2$, there are exactly two machines $i_1, i_2$
such that $x^*_{i_1 j} = x^*_{i_2 j} = 1/2$. Moreover, the processing times
on these machines must be finite (and hence at most $p_{\max}$). We define a vector $v^{(j)}\in \mathbb R^m$
with
\begin{equation*}
    v^{(j)}_i = \begin{cases}
        p_{i_1 j}/2 p_{\max} & \text{ if } i = i_1 \\
        - p_{i_2 j}/2 p_{\max} & \text{ if } i = i_2 \\
        0 & \text{ otherwise.}
    \end{cases}
\end{equation*}
Here the choice of which value is positive and which is negative
can be made arbitrarily. Notice that all vectors $v^{(j)}$ have $\ell_1$-norm at most~$1$.

Order the jobs in $J_2$ as $\{j_1,j_2,\dotsc,j_{|J_2|}\}$ in non-decreasing order of their release time,
Consider the corresponding vectors $v^{(j)}$ in this order and 
apply the prefix Beck-Fiala theorem.  
Let $\epsilon_j \in \{-1, 1\}$ be the signs
such that for all $k=1,2,\dotsc,|J_2|$, it holds that
\begin{equation}
 \lVert \epsilon_{j_1} v^{(j_1)} + \cdots + \epsilon_{j_k} v^{(j_k)} \rVert_\infty \le g(n, m) \ .
\end{equation}

We now show that this implies an integral solution of value
at most $T^* + 2 g(n, m) \cdot p_{\max}$.
Let $j\in J$. If $j$ is integrally assigned in $x^*$, we assign it to the same
machine. On the other hand if $j\in J_2$, we assign it depending on the choice
of $\epsilon_j$. Recall that $v^{(j)}_{i_1} = p_{i_1 j} / 2 p_{\max}$
and $v^{(j)}_{i_2 } = - p_{i_2 j} / 2 p_{\max}$ for the two machines $i_1, i_2$ with
$x^*_{i_1 j} = x^*_{i_2 j} = 1/2$. If $\epsilon_j = 1$, we assign $j$ to $i_1$, or to $i_2$ otherwise. We denote the resulting integral solution by $x$.
 By construction each job $j\in J_2$ and each machine $i$
  satisfy
  \begin{equation}
      x_{ij} p_{ij} = x^*_{ij} p_{ij} + p_{\max} \cdot \epsilon_j v^{(j)}_i \ .
  \end{equation}
  We will now verify the interval constraints. To this end, let $t_1 \le t_2$
  and $i\in M$. We calculate
  \begin{align*}
      \sum_{j\in J : r_j\in [t_1, t_2]} x_{ij} p_{ij} &= \sum_{j\in J \setminus J_2 : r_j\in [t_1, t_2]} x_{ij} p_{ij} + \sum_{j\in J_2 : r_j\in [t_1, t_2]} x_{ij} p_{ij} \\
      &= \sum_{j\in J \setminus J_2 : r_j\in [t_1, t_2]} x^*_{ij} p_{ij} + \sum_{j\in J_2 : r_j\in [t_1, t_2]} \left[ x^*_{ij} p_{ij} + p_{\max} \cdot \epsilon_j v^{(j)}_i \right] \\
      &\le t_2 - t_1 + T^* + p_{\max} \cdot \sum_{j\in J_2 : r_j\in [t_1, t_2]} \epsilon_j v^{(j)}_i \ .
  \end{align*}
  Let $\ell$ be largest index such that $r_{j_{\ell}} \le t_2$, and let $k$ be largest index such that
  $r_{j_{k}} < t_1$.
  Then we have
  \begin{align}
      \sum_{j\in J_2 : r_j\in [t_1, t_2]} \epsilon_j v^{(j)}_i
      &= \sum_{h = 1}^{\ell} \epsilon_{j_h} v^{(j_h)}_i - \sum_{h = 1}^{k} \epsilon_{j_h} v^{(j_h)}_i \label{eq:prefix-maxflow} \\
      &\le \bigg\lVert \sum_{h = 1}^{\ell} \epsilon_{j_h} v^{(j_h)} \bigg\rVert_\infty
      + \bigg\lVert \sum_{h = 1}^{k} \epsilon_{j_h} v^{(j_h)} \bigg\rVert_\infty \notag \\
      &\le 2 g(n, m) \ . \notag
  \end{align}
  We conclude that
  \begin{equation*}
      \sum_{j\in J : r_j\in [t_1, t_2]} x_{ij} p_{ij} \le t_2 - t_1 + T^* + 2 g(n, m) \cdot p_{\max}\, ,
  \end{equation*}
  for every interval $[t_1,t_2]$ and machine $i$, which implies that the max flow time for the solution $x$ is $T^* + 2 g(n, m) \cdot p_{\max}$.
  \end{proof}
  
Using the bound $g(n,m) = O(\sqrt{\log n})$ due to Banaszczyk~\cite{DBLP:journals/rsa/Banaszczyk12},
 Lemmas~\ref{lem:maxflow-reduction} and~\ref{lem:maxflow-halfintegral} imply
that the integrality gap of the assignment LP is at most $O(\sqrt{\log n})$.

\section{Total flow time}
\label{sec:total_flow_time}
For total flow time we follow a similar approach as for max flow time.
We start by reducing to half-integral solutions and then show how to round half-integral solutions.

However, for total flow time the LP formulation is more involved
and the rounding poses some serious obstacles due to the time indexed variables.
In particular, recall that for max flow time the release times gave a consistent ordering of the jobs and it was
sufficient to bound the load on each interval (w.r.t.~this order of jobs). On the other hand, for total flow time,
we will need bounds over all the jobs {\em assigned} to a particular time interval on a machine. However, as the time to which
a job is assigned may differ on each machine, at first glance there seems to be no single order
to execute the prefix Beck-Fiala approach. We circumvent this problem by carefully rearranging
jobs within certain groups so that ultimately we arrive at one consistent order across all machines.

\subsection{The time indexed LP.}
We now introduce the standard time-indexed LP for total flow time problems~\cite{DBLP:conf/icalp/GargK06, DBLP:conf/focs/GargK07, DBLP:conf/isaac/GargKM08, DBLP:conf/stoc/BansalK15}.
We assume that time is slotted and consider a large enough time horizon.
The variables $y_{ijt}$ describe whether machine $i$ processes job $j$
at time $t$.
\begin{align}
    \min \sum_{i\in M} \sum_{j\in J} \sum_{t\ge r_j}
    &\left( \frac{t - r_j}{p_{ij}} + \frac 1 2 \right) y_{ijt} \label{obj:tot-flow-1} \\
    \sum_{i\in M} \sum_{t \ge r_j} \frac{y_{ijt}}{p_{ij}} &= 1 & \forall j\in J \notag \\
    \sum_{j\in J : r_j \ge t} y_{ijt} &\le 1 & \forall i\in M, \quad \forall t \label{cons:capacity-1} \\
    y_{ijt} &\ge 0 \notag
\end{align}
The first set of constraints ensure that each job is processed completely, and the second set of constraints ensure that  each
machine can only process a volume of  at most $1$ at  each unit of time.

By standard discretization techniques we can assume that all processing times and release times
are integers and bounded by $\text{poly}(n)$, up to a negligible $1+o(1)$ factor loss in the approximation ratio, see e.g.~\cite{DBLP:conf/stoc/BansalK15}.
So the time horizon is $\text{poly}(n)$ and the LP also has polynomial size.
Moreover, one may assume that given an integral assignment of jobs of machines, the jobs on each machine are scheduled according to 
shortest-remaining-processing-time (SRPT), as it is the optimal scheduling policy for a single machine. 

It may not be immediate why the objective function forms a valid lower bound
for the total flow time of this solution, but this holds and we refer the
reader to~\cite{DBLP:conf/stoc/BansalK15} for details.

Bansal and Kulkarni~\cite{DBLP:conf/stoc/BansalK15} introduce a further
relaxation of the LP above by first grouping together for each machine $i$ the
jobs $j$ with $p_{ij} \in (2^{k-1}, 2^k]$. Note that there
are at most $\log P$ non-empty groups, where we recall that $P$ is
the ratio between maximum and minimum finite processing time.
Apart for a few differences that we will point out below,
they describe the following LP, which we will refer to as the \emph{auxiliary LP}.
\begin{align}
    \min \sum_{i\in M} \sum_{k} \sum_{j\in J : p_{ij} \in (2^{k-1}, 2^k]} \sum_{t\ge r_j}
    &\left( \frac{t - r_j}{2^k} + \frac 1 2 \right) y_{ijt} \label{eq:objective} \\
    \sum_{i\in M} \sum_{t \ge r_j} \frac{y_{ijt}}{p_{ij}} &= 1 & \forall j\in J \notag \\
    \sum_{j\in J : p_{ij} \le 2^k} \sum_{t \in [t_1, t_2]} y_{ijt} &\le t_2 - t_1 + \alpha \cdot 2^k & \forall i\in M \ \forall k \ \forall t_1 < t_2 \label{eq:capacity} \\
    y_{ijt} &\ge 0 \notag
\end{align}
Notice that we replaced the capacity constraint \eqref{cons:capacity-1} at time step $t$ in the original LP, by an aggregate capacity constraint over time intervals in \eqref{eq:capacity}.
Here $\alpha$ is a parameter that allows the capacity constraints~(\ref{eq:capacity})
to be slightly violated. We call a solution $\alpha$-relaxed, if it is feasible for $\alpha$.

It is easy to verify that any solution for the previous LP remains feasible
for this LP (with $\alpha = 0$) and the cost can only decrease.
This decrease is due to the slight change in the objective \eqref{eq:objective}, as
we replaced $p_{ij}$ in \eqref{obj:tot-flow-1} by the upper bound $2^k$.
This modification is not used in~\cite{DBLP:conf/stoc/BansalK15}, but has been
in some other related works~\cite{DBLP:conf/focs/GargK07, DBLP:conf/isaac/GargKM08}. It will allow us
to freely rearrange jobs of the same group without increasing
the cost. It should be noted that replacing $p_{ij}$ with $2^k$ can reduce
the cost by at most a factor of~$2$, which is negligible for our purposes.

A crucial and non-trivial ingredient of the proof in~\cite{DBLP:conf/stoc/BansalK15},
which we will reuse here, is that ``integral'' $\alpha$-relaxed solutions
can be used to construct a schedule of only slightly worse cost.
In this case integral means that for every job $j$ we have  $y_{ijt} = p_{ij}$ for exactly one machine $i$ and one particular time  $t$,
and zero everywhere else.
\begin{theorem}[Section 2.2 in~\cite{DBLP:conf/stoc/BansalK15}]
\label{thm:BK-schedule}
Let $y$ be an integral (so that $y_{ijt} \in \{0, p_{ij}\}$),
$\alpha$-relaxed solution to the auxiliary LP.
Then there is a schedule that assigns each jobs $j$ to the machine $i$
with $y_{ijt} = p_{ij}$ for some $t$, where the total flow time is at most
\begin{equation*}
    O(\alpha \log P \cdot \mathrm{LP}) \ .
\end{equation*}
Here $\mathrm{LP}$ is the cost of $y$ with respect to the objective (\ref{eq:objective}).
\end{theorem}
Although at first glance constraint~(\ref{eq:capacity}) looks similar to the capacity constraint
in max flow time and hence promising to apply prefix discrepancy,
there is a crucial difference.
Prefix discrepancy only allows us to obtain bounds on discrepancy of intervals with respect
to a single fixed order of jobs. In the case of max flow time this ordering was by release time.
In total flow time we need bounds with respect to the time the jobs are scheduled, and the problem is that this order may not be the same for each machine.

To circumvent this obstacle, we will use that an LP solution can be normalized
so that all jobs of the same group are scheduled in a globally consistent order.
\begin{definition}[Consistent order property]{\rm
Let $\prec$ be a order on the jobs with $j \prec j' \Rightarrow r_j \le r_{j'}$.
We say that a solution $y$ to the auxiliary LP has the consistent order property (w.r.t.~$\prec$), if
it satisfies for all $k, i, j, j'$ with $p_{i j}, p_{i j'} \in (2^{k-1}, 2^k]$:
\begin{equation*}
    \exists\, t < t' \text{ with } y_{i j t} > 0, y_{i j' t'} > 0 \quad \Rightarrow \quad j \preceq j' \ .
\end{equation*}
}\end{definition}
\begin{lemma}\label{lem:time-index-order}
For every $\alpha$-relaxed solution $y$ to the auxiliary LP, there exists a consistently ordered $\alpha$-relaxed solution $y'$,
which has the same cost and $\sum_{t \ge r_j} y'_{ijt} = \sum_{t \ge r_j} y_{ijt}$ for all jobs $j$ and machines $i$.
\end{lemma}
\begin{proof}
This follows from a simple exchange argument.
Suppose that for some $t < t'$ we have that $y_{ijt} > 0$ and $y_{i j' t'} > 0$, but $j' \prec j$.
In particular, it holds that $r_{j'} \le r_{j} \le t < t'$.
This implies that both jobs can be assigned to both times.
Let $\delta = \min\{y_{i j t}, y_{i j' t'}\}$.
We increase $y_{i j' t}$ by $\delta$ and reduce $y_{i j t}$ by the same amount.
Likewise, increase $y_{i j t'}$ by $\delta$ and reduce $y_{i j' t'}$ by it.

Clearly, the total amount of each job that is scheduled on $i$ does not change
and also the amount of processing of each group done by $i$
at each unit of time does not change.
Thus, the solution remains feasible and $\alpha$-relaxed. Further,
the coefficient of $y_{i j t}$ and $y_{i j' t}$ (also $y_{i j t'}$ and $y_{i j' t'}$)
are the same in the objective function, because the jobs belong to the
same group and we have rounded the denominator to $2^k$.
Consequently, also the solution value did not change.

As each such exchange reduces the number of violations, after finitely many exchanges, the condition of the lemma must hold.
\end{proof}

\subsection{Reducing to half-integral solutions}
We now show that it is sufficient to give a rounding algorithm
for half-integral solutions.
\begin{lemma}
\label{lem:totalflow-reduction}
Let $f(n, m)$ be a non-decreasing function such that for any
instance and any half-integral ($y^*_{ijt} \in \{0, p_{ij}/2, p_{ij}\}$),
$\alpha$-relaxed solution $y^*$ of the auxiliary LP
there exists an integral $(\alpha + f(n, m))$-relaxed solution $y$ of cost no greater than that of $y^*$.
Then the integrality gap of the time indexed LP is at most
\begin{equation*}
    O(f(n^2, m) \cdot \log P) \ .
\end{equation*}
\end{lemma}
\begin{proof}
Let $y^*$ be an optimal, $0$-relaxed solution to the auxiliary LP.
We start by augmenting $y^*$ slightly to obtain a $1$-relaxed solution
$y^{(\ell)}$ where each variable $y^{(\ell)}_{ijt}$ is an integer multiple
of $p_{ij} / 2^{\ell}$ for $\ell = \lceil \log_2 n \rceil$:
Let $j$ be a job and suppose not all its variables are integer multiples
of $p_{ij} / 2^{\ell}$. Since $\sum_{i,t} y^{(\ell)}_{ijt} / p_{ij} = 1$
there must be at least two such variables $y^{(\ell)}_{ijt}$ and $y^{(\ell)}_{i'jt'}$.
We either increase $y^{(\ell)}_{ijt}$ by $\delta p_{ij}$ and decrease
$y^{(\ell)}_{i'jt'}$ by $\delta p_{i'j}$ or vice versa, where we choose
$\delta > 0$ minimal such that one of the two variables becomes an integer multiple.
We select the direction that does not increase
the cost. This procedure reduces the number
of variables that are not integer multiples each iteration.
Therefore eventually all variables are integer multiples of $p_{ij} / 2^\ell$.
Notice also that each variable $y^{(\ell)}_{ijt}$
changes by at most $p_{ij} / 2^\ell \le p_{ij} / n$.
Thus, the left-hand side of~(\ref{eq:capacity})
increases by at most $n \cdot 2^k / 2^\ell \le 2^k$ and $y^{(\ell)}$ is $1$-relaxed.

We now argue inductively that for every $h = \ell,\ell-1,\dotsc,0$
there exists a $\alpha$-relaxed solution $y^{(h)}$, where each variable $y^{(h)}_{ijt}$
is an integer multiple of $p_{ij} / 2^{h}$, for
$\alpha = f(n^2, m) / 2^{h-1} + 1$. By construction of $y^{(\ell)}$
this holds for the base case.
Now assume we are given $y^{(h)}$ for some $h$. Using the rounding
of a half-integral solution, we will construct $y^{(h-1)}$.
To this end, we define a new instance of total flow time by splitting
each job into $2^{h - 1}$ equal pieces:
A job $j$ is split into jobs $j'$ with $p'_{ij'} = p_{ij} / 2^{h - 1}$
for all $i$. 

One can interpret $y^{(h)}$ as a half-integral solution to this new instance.
Similarly, an integral solution to the
new instance corresponds to a solution $y^{(h-1)}$ for the original
instance, where each $y^{(h-1)}_{ijt}$ is an integer multiple of
$p_{ij} / 2^{h-1}$. We use the rounding from the premise of the lemma
to produce this $y^{(h-1)}$. It remains to bound the increase
in the left-hand side of constraint~(\ref{eq:capacity}).

Consider a machine $i$ and a class $k$.
All jobs $j$ with $p_{ij} \in (2^{k-1}, 2^k]$ are split into
jobs $j'$ with $p'_{ij} \in (2^{k-h}, 2^{k - h + 1}]$.
Hence, the rounding will increase the right-hand side of~(\ref{eq:capacity})
for this group only by $f(n^2, m) \cdot 2^{k - h + 1}$. Here we
use that the number of jobs in the new instance is at most $2^{\ell-1} n \le n^2$.
It follows that
\begin{align*}
    \sum_{j: p_{ij}\le 2^k} \sum_{t\in [t_1, t_2]} y^{(h - 1)}_{ijt} 
    &\le t_2 - t_1 + 2^{k - h + 1} f(n^2, m) + 2^k + f(n^2, m) \cdot 2^{k - h + 1} \\
    & \le t_2 - t_1 + 2^{k - h + 2} f(n^2, m) + 2^k \ .
\end{align*}
Applying Theorem~\ref{thm:BK-schedule} to $y^{(0)}$ finishes the proof.
\end{proof}
\subsection{Rounding half-integral solutions}
In this section we describe how to round half-integral solutions.
\begin{lemma}
\label{lem:totalflow-halfintegral}
Let $g(n, m)$ be a non-decreasing bound on prefix Beck-Fiala.
Further, let $y^*$ be a half-integral
(with $y^*_{ijt} \in \{0, p_{ij} / 2, p_{ij}\}$),
$\alpha$-relaxed solution for the auxiliary LP.
Then there is an integral $(\alpha + O(g(n, m \log P)))$-relaxed integral solution.
\end{lemma}
\begin{proof}
We sort the jobs by release time and break ties arbitrarily.
Let $\prec$ denote the resulting total order.
By Lemma~\ref{lem:time-index-order} there is a consistently ordered (w.r.t.~$\prec$) solution $\bar y$ with
$\sum_{t \ge r_j} \bar y_{ijt} \in \{0, p_{ij} / 2, p_{ij}\}$ for every job~$j$ and machine~$i$.
Although the time assignments might not be half-integral in $\bar y$, the assignments to machines are.

We construct an instance of the prefix discrepancy as follows.
Denote by $J_2$ all jobs which are partially assigned to two different machines. Let $j\in J_2$ and $i_1,i_2$ be the machines
with
$\sum_{t \ge r_j} \bar y_{i_1 j t} = p_{i_1 j} / 2$
and $\sum_{t \ge r_j} \bar y_{i_2 j t} = p_{i_2 j} / 2$.
Define $v^{(j)} \in \mathbb R^{m \times \log P}$ with
\begin{equation*}
    v^{(j)}_{ik} = \begin{cases}
      p_{ij}/2^{k+1} &\text{ if } i = i_1 \text{ and } p_{i_1 j} \in (2^{k-1}, 2^k] \\
      - p_{ij}/2^{k+1} &\text{ if } i = i_2 \text{ and } p_{i_2 j} \in (2^{k-1}, 2^k] \\
      0 &\text{ otherwise. }
    \end{cases}
\end{equation*}
Here, the choice which component is positive and which is negative
is made arbitrarily. We order these vectors by $\prec$
and obtain signs $\epsilon_j$ with a prefix bound of~$g(n, m)$.
Next, we elaborate how to derive the integral solution~$y$.

For every job $j\notin J_2$, we have that $j$ is completely assigned to a single machine $i$, but potentially to different times.
We will assign it to the earliest time it is assigned to in $\bar y$.
Formally, let $t$ be minimal with $\bar y_{ijt} > 0$.
We set $y_{ijt}$ to $p_{ij}$ and all other variables for $j$ to zero.

Next consider a job $j\in J_2$. The job has two potential machines to be assigned to and the sign $\epsilon_j$
determines which one we choose.
More precisely, let $i_1, i_2$ be the machines for which
$\sum_{t \ge r_j} \bar y_{i_1 j t} = p_{i_1 j} / 2$
and $\sum_{t \ge r_j} \bar y_{i_2 j t_2} = p_{i_2 j} / 2$.
Further, let $t_1$ and $t_2$ be the earliest times $j$ is assigned to on machine $i_1$ and $i_2$.
If $\epsilon_j v^{(j)}_{i_1 k} > 0$ for the corresponding $k$, we set $y_{i_1 j t_1} = p_{i_1 j}$. Otherwise,
set $y_{i_2 j t_2} = p_{i_2 j}$. Again, all other variables for $j$ are set to zero.
The solution $y$ then satisfies for all $i$, $k$, $j_1 \preceq j_2$ that
\begin{align}
    \sum_{\substack{j: j_1 \preceq j \preceq j_2 \\ p_{ij} \in (2^{k-1}, 2^k]}} \sum_{t \ge r_j} y_{ijt}
    &= \sum_{\substack{j\notin J_2: j_1 \preceq j \preceq j_2 \\ p_{ij} \in (2^{k-1}, 2^k]}} \sum_{t \ge r_j} y_{ijt}
    + \sum_{\substack{j\in J_2: \preceq j \preceq j_2 \\ p_{ij} \in (2^{k-1}, 2^k]}} \sum_{t \ge r_j} y_{ijt} \notag \\
    &= \sum_{\substack{j\notin J_2: j_1 \preceq j \preceq j_2 \\ p_{ij} \in (2^{k-1}, 2^k]}} \sum_{t \ge r_j} \bar y_{ijt}
    + \sum_{\substack{j\in J_2: j_1 \preceq j \preceq j_2 \\ p_{ij} \in (2^{k-1}, 2^k]}} \bigg[ 2^k \epsilon_j v^{(j)}_{ik} + \sum_{t \ge r_j} \bar y_{ijt} \bigg] \notag \\
    &= \sum_{\substack{j: j_1 \preceq j \preceq j_2 \\ p_{ij} \in (2^{k-1}, 2^k]}} \sum_{t \ge r_j} \bar y_{ijt}
    + \sum_{j\in J_2: j \preceq j_2} 2^k \epsilon_j v^{(j)}_{ik}
    - \sum_{j\in J_2: j \prec j_1} 2^k \epsilon_j v^{(j)}_{ik} \notag \\
    &\le \sum_{\substack{j: j_1 \preceq j \preceq j_2 \\ p_{ij} \in (2^{k-1}, 2^k]}} \sum_{t \ge r_j} \bar y_{ijt} + 2 g(n, m) \cdot 2^k \ . \label{eq:prefix-totalflow}
\end{align}
It remains to relate this bound to the constraints~(\ref{eq:capacity}).
To this end consider some $k,i,t_1 < t_2$.
Furthermore, denote by $j_1$ the smallest (by $\prec$) job of group $k$
with $\bar y_{ijt} > 0$ for some $t\in [t_1, t_2]$.
Likewise, let $j_2$ denote the largest such job.
By the consistent order property of $y$,
no job $j$ of group $k$ with $j_1 \prec j \prec j_2$ is processed outside $[t_1, t_2]$ on $i$, in particular,
\begin{equation*}
    \sum_{j: p_{ij} \in (2^{k-1}, 2^k]} \sum_{t \in [t_1, t_2]} \bar y_{ijt} \ge \sum_{\substack{j: j_1 \prec j \prec j_2 \\ p_{ij} \in (2^{k-1}, 2^k]}} \sum_{t \ge r_j} \bar y_{ijt} \ge \sum_{\substack{j: j_1 \preceq j \preceq j_2 \\ p_{ij} \in (2^{k-1}, 2^k]}} \sum_{t \ge r_j} \bar y_{ijt} - 2 \cdot 2^k \ .
\end{equation*}
Furthermore, all jobs $j$ with $y_{ijt} > 0$ for some $t\in [t_1, t_2]$ must also satisfy
$\bar y_{ijt} > 0$ and, in particular, $j_1 \preceq j \preceq j_2$.
This implies
\begin{align*}
    \sum_{j : p_{ij} \in (2^{k-1}, 2^k]} \sum_{t \in [t_1, t_2]} y_{ijt}
    &\le \sum_{\substack{j: j_1 \preceq j \preceq j_2 \\ p_{ij} \in (2^{k-1}, 2^k]}} \sum_{t \ge r_j} y_{ijt} \\
    &\le \sum_{\substack{j: j_1 \preceq j \preceq j_2 \\ p_{ij} \in (2^{k-1}, 2^k]}} \sum_{t \ge r_j} \bar y_{ijt} + 2 g(n, m) \cdot 2^k \\
    &\le \sum_{j : p_{ij} \in (2^{k-1}, 2^k]} \sum_{t \in [t_1, t_2]} \bar y_{ijt} + 2 g(n, m) \cdot 2^k + 2 \cdot 2^k \ .
\end{align*}
Since this holds for any $k$, summing over all $k' \le k$ yields
\begin{align*}
    \sum_{j : p_{ij} \le 2^k} \sum_{t \in [t_1, t_2]} y_{ijt}
    &\le \sum_{j : p_{ij} \le 2^k} \sum_{t \in [t_1, t_2]} \bar y_{ijt} + O(g(n, m) \cdot 2^k) \\
    &\le t_2 - t_1 + \alpha \cdot 2^k + O(g(n, m) \cdot 2^k) \ .
\end{align*}
Notice that the cost of the solution $y$ may be higher than that of $\bar y$ (equivalently, that of $y^*$),
but in this case we can perform the same construction with $-\epsilon_j$ (all signs flipped). Denote this solution
by $y'$.
Consider now the solution $y''$ obtained from $\bar y$ where for all jobs $j$ and machines $i$ we set
$y''_{ijt_1} = \sum_{t \ge r_j} \bar y_{ijt}$ for the earliest time $t_1$ with $\bar y_{ijt_1} > 0$.
All other variables are set to zero.
Clearly, the cost of $y''$ is not more than that of $\bar y$, since we only moved jobs to earlier times.
By construction of $y$ and $y'$ we have that $y'' = 1/2 \cdot y + 1/2 \cdot y'$.
Thus, the cost of one of the two solutions has to be lower than that of $\bar y$.
We conclude that there is a $\alpha'$-relaxed integral solution for $\alpha' = \alpha + O(g(n, m))$.
\end{proof}
Applying Banaszczyk's bound on prefix Beck-Fiala~\cite{DBLP:journals/rsa/Banaszczyk12}, we may choose $g(n, m) = O(\sqrt{\log n})$ and
Lemmas~\ref{lem:totalflow-reduction} and~\ref{lem:totalflow-halfintegral} 
together imply that the integrality gap
of the time indexed LP is at most $O(\sqrt{\log n} \log P)$.

\section{Equivalence of max flow time and $2$-sparse prefix discrepancy}
\label{sec:relation}
In this section we make the specific form of the prefix discrepancy problem
in the previous proofs explicit.
Namely, we show an equivalence between the integrality gap of max flow time
and a one-sided discrepancy bound on $2$-sparse vectors with $\ell_1$-norm
bounded by~$1$.
This implies a non-trivial relation between the
integrality gaps of the assignment LP for max flow time and
the time indexed LP for total flow time. In addition, it motivates the study of prefix Beck-Fiala in the seemingly easy case of $2$-sparse vectors (Conjecture~\ref{conj:2-sparse}). We further explore techniques for that special case in Section~\ref{sec:game}.

Let $V_m \subseteq \mathbb R^m$ denote the set of all vectors $v$ of the form
\begin{equation*}
    v_i = \begin{cases}
      p_1 &\text{ if } i = i_1 \\
      -p_2 &\text{ if } i = i_2 \\
      0 &\text{ otherwise.}
    \end{cases}
\end{equation*}
where  $p_1, p_2\in [0, 1/2]$ and $i_1,i_2\in [m]$ with $i_1 \neq i_2$.
Let $g(n, m)$ denote the infimum over all $C$
such that for every sequence of vectors $v^{(1)},v^{(2)},\dotsc,v^{(n)} \in V_m$
there are signs $\epsilon_1,\epsilon_2,\dotsc,\epsilon_n$
with
\begin{equation*}
    \epsilon_k v_i^{(k)} + \epsilon_{k+1} v_i^{(k+1)} + \cdots + \epsilon_\ell v_i^{(\ell)} \le C
\end{equation*}
for every $1 \le k < \ell \le m$ and $1 \le i \le m$.
Since vectors in $V_m$ have $\ell_1$-norm at most~$1$ and we can obtain
bounds on all intervals by subtracting two prefixes, $g(n, m)$
is bounded by two times the prefix Beck-Fiala discrepancy (when considering $2$-sparse vectors).
We consider here the special form with intervals instead of prefixes,
because, as we will show, it is
equal (up to constants) to the integrality gap of max flow time.
Moreover, it also forms an upper bound on total flow time.

\begin{theorem}
  The integrality gap of the assignment LP for max flow time is at
  most $O(g(n, m))$ and that of the time indexed LP for
  total flow time is at most $O(g(n, m \log P) \cdot \log P)$.
\end{theorem}
\begin{proof}
It is easily checked that in the proofs, which involve prefix Beck-Fiala,
Lemmas~\ref{lem:maxflow-halfintegral} and~\ref{lem:totalflow-halfintegral},
all the vectors we construct lie in $V_m$ and $V_{m \log P}$.
Moreover, in all our uses of the prefix Beck-Fiala bound,
namely~(\ref{eq:prefix-maxflow}) and~(\ref{eq:prefix-totalflow}), we only require upper
bounds on each interval and dimension.
Hence, the same proofs can be used to show
that the integrality gap of max flow time is
at most $O(g(n, m))$ and that of total flow time is at most
$O(g(n, m \log P) \cdot \log P)$.
\end{proof}

For max flow time the other direction also holds. 
\begin{theorem}
  Let $f(n, m)$ be an upper bound on the integrality gap of the assignment LP for the max flow time problem. Then $g(n, m) \le f(n(m+1), m)$.
\end{theorem}
\begin{proof}
Given a sequence of vectors $v^{(1)},v^{(2)},\dotsc,v^{(n)} \in V_m$
we construct an instance of max flow time as follows.
There are $m$ machines and
for each time $t = 1,2,\dotsc,n$ there are $m+1$ jobs released at this time.
These are derived from $v^{(t)}$.
Let $i_1, i_2, p_1, p_2$ be the defining values of $v^{(t)}$.
We introduce one special job $j_t$, which has $p_{i_1 j_t} = 2 p_1$,
$p_{i_2 j_t} = 2 p_2$, and $p_{ij_t} = \infty$ for all $i\notin \{i_1,i_2\}$.
Then for each machine $i$ we add one job $j$:
If $i = i_1$ we set $p_{ij} = 1 - p_1$ and $p_{i'j} = \infty$ for all $i' \neq i$.
Likewise,
if $i = i_2$ we set $p_{ij} = 1 - p_2$ and $p_{i'j} = \infty$ for all $i' \neq i$.
Finally, for all $i\notin \{i_1, i_2\}$ we set $p_{ij} = 1$ and $p_{i'j} = \infty$
for all $i'\neq i$.
Notice that by assigning a processing time of $\infty$ to all but one machine,
this fixes the assignment of the job.
With this construction, the optimum of the assignment LP is $1$:
It cannot be smaller than $1$, since some jobs have processing time at least
$1$ on all machines.
On the other hand, assigning one half of each of the special jobs $j_t$
to each of the two machines it has finite processing times on
produces a schedule, where each machine is assigned a volume of exactly $1$
from jobs released at each time $t = 1,2,\dotsc,n$.
Hence, the interval constraints of the assignment LP are satisfied
with $T = 1$.
Now let $\OPT$ denote the value of the optimal integral solution $x$.
We now construct signs as follows.
Consider some $v^{(t)}$, $t \in \{1,2,\dotsc,n\}$ and let
$i_1, i_2, p_1, p_2$ be its defining values. In the integral solution
job $j_t$ must be assigned to either $i_1$ or $i_2$. If $x_{i_1 i_t} = 1$,
we set $\epsilon_t = 1$. Otherwise, $\epsilon_t = -1$.
It follows that
\begin{equation*}
    \sum_{j\in J : r_j = t} x_{ij} p_{ij} - 1 = \epsilon_t v_i^{(t)} \ .
\end{equation*}
We conclude that for all $k < \ell$ and all $i$ it holds that
\begin{align*}
    \epsilon_k v_i^{(k)} + \epsilon_{k+1} v_i^{(k+1)} + \cdots + \epsilon_\ell v_i^{(\ell)}
    & = \sum_{j : r_j \in [k, \ell]} x_{ij} p_{ij} - (\ell - k) \\
    & \le (\ell - k) + \OPT - (\ell - k) = \OPT \\
    & \le f(n (m+1), m) \ .
\end{align*}
Thus, $g(n, m) \le f(n (m+1), m)$.
\end{proof}
The two theorems above imply Theorem~\ref{thm:relation}:
Let $f(n, m)$ be an upper bound on the integrality gap of the assignment
LP for the max flow time problem. Then $O(f(n(m+1), m\log P) \cdot \log P)$ upper bounds the integrality gap of the time indexed LP for the total flow time problem.
As our bounds satisfy $f(n, m) \le O(\sqrt{\log n})$, the polynomial increase in the parameters is negligible. In particular, a constant bound for max flow time would imply a tight (up to constants) bound of $O(\log P)$ for total flow time.

\section{Discrepancy game}
\label{sec:game}
In this section we consider a two player game, which is closely related to the prefix discrepancy of 2-sparse vectors and other related discrepancy questions.
Recall that flow time can be reduced to prefix discrepancy
of vectors that have only two non-zero entries (2-sparse),
see Section~\ref{sec:relation}.
There is an approach due to Spencer~\cite{spencer1994ten}
that works very well in the case of 2-sparse
vectors with only entries $\{-1, 0, 1\}$.
Hence, it is tempting to try and generalize
this approach to fractional vectors to make progress on flow time.
In this section we make a case that this is not likely to succeed.
We show that Spencer's approach naturally extends to a certain
two player game, which can be used to prove discrepancy bounds.
However, we show that although there is a good strategy for
the game in the $\{-1, 0, 1\}$ setting, with fractional
values there is none.

\begin{definition}[One dimensional discrepancy game]{\rm
Let $v^{(1)},v^{(2)},\dotsc,v^{(n)} \in [-1, 1]$.
Initially set $\epsilon_i = 0$ for all $i$.
There are two players, the \emph{maker} and the \emph{breaker}.
One of the players is determined to start and chooses an uncolored element $\epsilon_i$ and colors it to
 $-1$ or $1$. Then the other player does the same.
This continues until all elements are colored.
The goal of the breaker is to construct a prefix $1,2,\dotsc,\ell$ such that $|v^{(1)} + v^{(2)} + \cdots + v^{(\ell)}|$
is large. At the same time the maker tries to keep all prefixes balanced.
}\end{definition}
We call a strategy \emph{robust}, if it also works against a player that is allowed
to wait, that is, to not color any element on their turn.
In the following we show that a robust strategy for the maker implies
bounds for the prefix discrepancy
of 2-sparse vectors and the weighted variant of two permutation discrepancy.
Further, in a similar way as a proof by Spencer for two permutation discrepancy~\cite{spencer1994ten},
one can
devise a robust, constant discrepancy strategy for the maker if all values are equal to $1$.
This implies constant bounds for the two aforementioned discrepancy questions, if all coefficients
are in $\{-1, 0, 1\}$. It therefore appears natural to try to generalize this game approach.
However, we will show that in the case of general weights, the game behaves differently
and that there exist instances, where the breaker can achieve an unbounded discrepancy.
This gives evidence that the weighted variants of
these discrepancy questions are much harder than the unweighted ones.

\subsection{A strategy when all values are one}

We show that the maker always has a good robust strategy when values are integral.
\begin{lemma}
In the one dimensional discrepancy game with values $\{-1,1\}$
there is a robust strategy for the maker to ensure a discrepancy of at most $4$ on all prefixes.
\end{lemma}
\begin{proof}
In the case that all values equal to $1$, we first assume the breaker starts
and $n$ is even to elaborate the main idea.
The maker has a simple strategy to ensure that all prefixes
are bounded by $1$:
We form pairs $(v^{(2i-1)}, v^{(2i)})$ for $i=1,2,\dotsc,n/2$ and ensure that
$\epsilon_{2i-1} + \epsilon_{2i} = 0$. More precisely, when the breaker colors one element of a pair,
the maker colors the other element in the opposite way.

We now extend this strategy to the case, where the maker starts or the breaker is allowed to wait.
First, note that we can still assume that $n$ is even by ignoring the last element and if the breaker
colors it, we treat this as a waiting move.
Ignoring the last element will add at most $1$
to the discrepancy bound we obtain.
We form the same pairs as before. When it is the maker's turn and there is some pair of which
exactly one element is colored, again we color the other element in the opposite way
(if there are multiple such pairs, we take the last one).
Otherwise, we take the first uncolored element and color it greedily.
More formally, if $i$ is the first uncolored element
and $\epsilon_1 v^{(1)} + \cdots + \epsilon_{i-1} v^{(i-1)}$ is negative, we set $\epsilon_i = 1$.
If it is non-negative, we set $\epsilon_i = -1$.

In order to bound the discrepancy in this strategy, we consider separately those pairs for which
the first element was colored due to the greedy strategy and all other pairs.
Let $(i_1, i_1 + 1), (i_2, i_2 + 1), \dotsc, (i_\ell, i_\ell + 1)$ denote pairs of the former kind
in correct order and let
$(j_1, j_1 + 1), (j_2, j_2 + 1), \dotsc, (j_k, j_k + 1)$ be the other indices.
It is not hard to see that $i_1, i_1 + 1, i_2, i_2 + 1,\dotsc$ were colored in exactly this order
and $i_1,i_2,\dotsc$ were chosen in the greedy way. It follows inductively that the prefixes of 
$v^{(i_1)},v^{(i_1+1)},v^{(i_2)},v^{(i_2+1)},\dotsc$ cannot exceed $2$ in absolute value:
Suppose first that the prefix up to $v^{(i_t + 1)}$ is $0$, $1$, or $2$.
Then $\epsilon_{i_{t+1}}$ was chosen as $-1$ and the prefix up to $v^{(i_{t+1} + 1)}$ cannot be bigger,
so it is in $\{-2,-1,0,1,2\}$. If the prefix is $-1$ or $-2$, this also follows in the same way.
Moreover, the prefixes of $v^{(j_1)},v^{(j_1+1)},v^{(j_2)},v^{(j_2+1)},\dotsc$ cannot exceed $1$
in absolute value as argued in the simpler case.
Adding the bounds for both subsequences and the loss due to making $n$ even, we conclude
that the strategy maintains a discrepancy bound of $4$ for all prefixes.

This approach easily extends to the case where values are $-1$ or $1$ by simply inverting the coloring of all $-1$ elements.
\end{proof}

\subsection{Reduction to the game}
In this section we show that prefix discrepancy and two permutation discrepancy
can be reduced to the game.
\begin{theorem}
Assuming the maker has a robust strategy to maintain a bound of $C$ on the discrepancy of all prefixes,
the prefix discrepancy of 2-sparse vectors is at most $2C$.
\end{theorem}
\begin{proof}
Let $v^{(1)}, v^{(2)},\dotsc,v^{(n)}\in\mathbb R^m$ be 2-sparse vectors in an instance of prefix discrepancy.
For all $i$ we define vectors $u^{(i)}, w^{(i)}$, where $u^{(i)}$ equals $v^{(i)}$ on the first
non-zero entry of $v^{(i)}$ and is zero everywhere else.
Likewise, $w^{(i)}$ contains only the second non-zero entry of $v^{(i)}$.
Let $\epsilon_1,\epsilon_2,\dotsc,\epsilon_n$ be signs. Then for any $k = 1,2,\dotsc,n$ it holds that
\begin{align*}
    \lVert \epsilon_1 v^{(1)} + \cdots \epsilon_k v^{(k)} \rVert_\infty
    &= \lVert \epsilon_1 (u^{(1)} + w^{(1)}) + \cdots \epsilon_k (u^{(k)} + w^{(k)}) \rVert_\infty \\
    &\le \lVert \epsilon_1 u^{(1)} + \cdots \epsilon_k u^{(k)} \rVert_\infty
    + \lVert \epsilon_1 w^{(1)} + \cdots \epsilon_k w^{(k)} \rVert_\infty \ .
\end{align*}
In other words, it suffices to find a coloring which is good simultaneously for $u^{(i)}$ and $w^{(i)}$.
Now suppose that player~$1$ wants to ensure that the prefix sums with respect to $u^{(i)}$
are good and player~$2$ does the same for $w^{(i)}$. We start with an empty coloring and then in each
round each player gets to color one element.
We can reduce this to the robust one dimensional discrepancy game:
Because of the symmetry, we will focus only on a strategy for player~$1$.
We run a given one dimensional strategy in parallel for each of the $m$ dimensions.
The input for each dimension are the vectors $u^{(i)}$ whose non-zero entry is in this dimension.
When player~$2$ fixes the sign of some $u^{(i)}$, we consult the strategy for the dimension
in which its non-zero entry lies and color another element in the same dimension.
This will ensure that the prefixes in every dimension remain bounded.
More precisely, the prefix discrepancy of $v^{(1)}, v^{(2)},\dotsc,v^{(n)}$ is at most $2 C$.
\end{proof}

In a similar way, we can reduce two permutation discrepancy to the same game.
\begin{definition}[Weighted two permutation discrepancy]{\rm
Let $v^{(1)},v^{(2)},\dotsc,v^{(n)} \in [-1, 1]$ and let $\sigma$ be a permutation on the $n$ elements.
The weighted two permutation discrepancy is concerned with the best bound $C_n$ such that
there exist signs $\epsilon_1,\epsilon_2,\dotsc,\epsilon_n$ with
\begin{align*}
    |\epsilon_1 v^{(1)} + \epsilon_2 v^{(2)} + \cdots + \epsilon_k v^{(k)} | &\le C_n \text{ and} \\
    |\epsilon_{\sigma(1)} v^{(\sigma(1))} + \epsilon_{\sigma(2)} v^{(\sigma(2))} + \cdots + \epsilon_{\sigma(k)} v^{(\sigma(k))} | &\le C_n
\end{align*}
for all $k = 1,2,\dotsc,n$.
}\end{definition}
\begin{theorem}
Assuming the maker has a robust strategy to maintain a bound of $C$ on the discrepancy of all prefixes,
the weighted two permutation discrepancy is at most $C$.
\end{theorem}
\begin{proof}
In a similar way as before, we will use the one dimensional discrepancy game to derive bounds
on this discrepancy measure.
To this end, player~$1$ needs to ensure that the prefixes with respect to $v^{(i)}$
are bounded and player~$2$ does the same with respect to $v^{(\sigma(i))}$.
Any bound on the discrepancy for the maker's strategy directly implies the same bound on the
weighted two permutation discrepancy.
\end{proof}

In particular, given the strategy for the case of all values equal to $1$ (or $-1$), we
get a bound of $8$ for the prefix discrepancy of $2$-sparse vectors and a bound of $4$
for the two permutation discrepancy when all entries are in $\{-1, 0, 1\}$.
\subsection{Hardness of the game.}

We will now show that the discrepancy game where the values are not necessarily integers is much harder.
\begin{lemma}
There is an
instance of the one dimensional discrepancy game
such that the breaker has a strategy to construct
a prefix with discrepancy $\Omega(\log(n) / \log\log(n))$.
\end{lemma}
\begin{proof}
Let $k$ be an even number. We recursively construct an instance as follows.
The instance $I_0$ contains no elements. Then for $i \ge 1$ let the elements in $I_i$ be
\begin{equation*}
    0.5 + \frac i k, \underbrace{I_{i-1}, I_{i-1}, \dotsc, I_{i-1}}_{k^2 \text{ times}}
\end{equation*}
We consider now the instance $I = I_{k / 2}$, which has a total of $k^{O(k)}$ elements,
each of which has a value between $0.5$ and $1$ and we will show that the breaker can construct a
prefix with discrepancy $\Omega(k)$.

Another way to arrive at this construction is to consider a complete $k^2$-ary tree of height~$k / 2$.
The value of a node in the $i$th layer is $1 - i/k$. For example,
the root is in layer~$0$ and therefore has value~$1$.
We arrive at the same sequence as before, if perform a preorder walk:
starting at the root, output the current node;
then recurse on the first child (and its subtree), then the second child, etc.

Based on the tree analogy, for an index $i$ we call $\text{next-sib}(i)$ the element
corresponding to the next child of $i$'s parent. In other words,
it is the next element in the sequence that has the same value as $v^{(i)}$, except when there
is a larger element between them, in which case it is undefined.
Further, we define $\text{first-child}(i)$ as the element corresponding to the first child of $i$.
Notice that by the order the tree is traversed, we have $\text{first-child}(i) := i+1$.
The $\ell$th child of $i$ is
$\text{next-sib}(\text{next-sib}( \cdots \text{next-sib}(\text{first-child}(i))\cdots))$, where $\text{next-sib}$
is nested $\ell-1$ times.

It suffices for the breaker to produce a large discrepancy on some prefix at some point of time in the game, even
if not all elements have been colored, yet. Then, for the rest of the game the breaker can maintain this large discrepancy by
always coloring the largest remaining element in this prefix so as to increase the discrepancy. From there on the maker can
reduce it by at most $1$ (if it is the maker's turn).
Our approach is to produce an interval of elements, whose
value is very large, that is,
$\epsilon_{k}v^{(k)} + \epsilon_{k+1}v^{(k+1)} + \cdots + \epsilon_{\ell}v^{(\ell)} \ge \Omega(k)$.
Then either the prefix up to $k-1$ has a large discrepancy or the one
up to $\ell$ does.

To this end, the breaker constructs a structure of the following form. There are indices $i_0, i_1, i_2\dotsc,i_{\ell + 1}$
with the following properties.
\begin{enumerate}
    \item Elements $i_1, i_2, \dotsc, i_{\ell}$ are all colored with $1$.
    Note that $i_0$ and $i_{\ell+1}$ are not necessarily colored with~$1$.
    \item For all $j = 0,1,\dotsc,\ell$ the uncolored elements in $v^{(i_j+1)},v^{(i_j+2)},\dotsc,v^{(i_{j+1}-1)}$ are strictly smaller than $v^{(i_j)}$.
    \item For all $j = 0,1,\dotsc,\ell$ we have that
    $u_j := \epsilon_{i_j+1}v^{(i_j+1)} + \cdots + \epsilon_{i_{j+1}-1}v^{(i_{j+1}-1)} \ge 0$.
    \item The elements in the subtree of $i_\ell$ (except for $i_\ell$ itself) are all uncolored.
    The same holds for the subtree of $\text{next-sib}(i_\ell)$, $\text{next-sib}(\text{next-sib}(i_\ell))$, etc., this time including the elements themselves.
    \item We have $v^{(i_{\ell+1})} \ge 1 - \ell / k$ (that is, $i_{\ell+1}$ is in layer at most $\ell$) and if $i_{\ell + 1}$ is a $j$th child, then for $u := u_0 + u_1 + \cdots + u_\ell$ we have $u \ge (j - 2) / k$.
\end{enumerate}
As will be shown, the breaker can maintain this structure as long as
$\text{next-sib}(i_{\ell + 1})$ and $\text{first-child}(i_\ell)$ are defined.
Notice that once either of them is not (which has to happen eventually)
the interval $[i_0 + 1, i_{\ell + 1} - 1]$ has a value of $\Omega(k)$:
If $\text{next-sib}(i_{\ell + 1})$ is not defined, then $i_{\ell+1}$ must be a $k^2$th child
and by the last property $u \ge \Omega(k)$. Thus
\begin{align*}
    \epsilon_{i_0 + 1} v^{(i_0 + 1)} + \epsilon_{i_0 + 2} v^{(i_0 + 2)} + \cdots + \epsilon_{ i_{\ell + 1} - 1} v^{( i_{\ell + 1} - 1)}
    &= u + \epsilon_{i_1} v^{(i_1)} + \epsilon_{i_2} v^{(i_2)} + \cdots + \epsilon_{i_\ell} v^{(i_\ell)} \\
    &\ge u \ge \Omega(k) \ .
\end{align*}
If on the other hand  $\text{first-child}(i_{\ell})$ is not defined, then $i_{\ell}$ is in
layer $k/2$ and so is $i_{\ell+1}$. Thus, by the last property $\ell \ge k / 2$
By the third property we have $u \ge 0$ and therefore using the first property we conclude
\begin{align*}
    \epsilon_{i_0 + 1} v^{(i_0 + 1)} + \epsilon_{i_0 + 2} v^{(i_0 + 2)} + \cdots + \epsilon_{ i_{\ell + 1} - 1} v^{( i_{\ell + 1} - 1)}
    &= u + \epsilon_{i_1} v^{(i_1)} + \epsilon_{i_2} v^{(i_2)} + \cdots + \epsilon_{i_\ell} v^{(i_\ell)} \\
    &= u + \frac 1 2 \ell \ge \Omega(k) \ .
\end{align*}
The breaker starts with $\ell = 1$, $i_0 = 1$ (the root), $i_1 = \text{first-child}(i_0)$ and $i_2 = \text{next-sib}(i_1)$ and colors $i_1$ with~$1$.

\paragraph*{Case 1: The maker colors an element inside an interval $[i_{t} + 1, i_{t+1} - 1]$ for some $t \ge 1$.}
We construct a new sequence $i'_0,i'_1,\dotsc,i'_{\ell+1}$
by merging the interval $[i_{t} + 1, i_{t+1} - 1]$ with that to its left and
append $\text{next-sib}(i_{\ell+1})$ to the structure.
Formally,
set $i'_j = i_j$ for $j = 0,1,\dotsc,t-1$;
set $i'_j = i_{j+1}$ for $j = t,t+1,\dotsc,\ell$;
$i'_{\ell+1} = \text{next-sib}(i_{\ell+1})$
and color $i'_\ell$ with $1$.
The element $i_t$ is removed from the special indices and is used to compensate for the element that
the maker has colored. Notice that by Property~2, $v^{(i_t)}$ is strictly larger and therefore $u$
increases by at least $1/k$. This is necessary because by moving $i'_{\ell+1}$ to the next sibling
the requirement on $u$ by Property~5 increases.

\paragraph*{Case 2: The maker colors no element inside the interval $[i_1 + 1, i_{\ell+1} - 1]$.}
We again construct a new sequence $i'_0,i'_1,\dotsc,i'_{\ell'+1}$,
this time increasing $\ell'$ to $\ell+1$.
We set $i'_0 = i_1 - 1$; $i'_j = i_j$ for $j = 1,2,\dotsc,\ell$;
$i'_{\ell+1} := \text{first-child}(i_\ell)$ and color it with $1$;
finally set $i'_{\ell + 2} = \text{next-sib}(i'_{\ell+1})$ and $\ell' = \ell + 1$.
Intuitively, in this case the maker does not interfere with our structure. Therefore,
we can extend it by increasing $\ell$. However, we need to be a bit careful.
Since it might be that the maker has colored an element in $[i_0+1, i_1 - 1]$, we increase
$i_0$ to $i_1 - 1$. This is to ensure that $u_0$ does not become negative.
Because the maker might have colored $i_{\ell+1}$, we discard this element from our structure
and replace it by $\text{first-child}(i_\ell)$. Finally, we add a new last element
as the next sibling. Notice that since $i'_{\ell'+1}$ is a second child, the bound on $u$ from
Property~$5$ is implied by Property~$3$, which is important since moving $i_0$ may have decreased $u$.

\end{proof}

\section{SDP relaxation for prefix Koml\'os}
\label{sec:SDP}

Let $v^{(1)},\ldots,v^{(n)} \in \R^m$ be arbitrary vectors satisfying $\|v^{(j)}\|_2\leq 1$ for all $j = 1,2,\dotsc,n$.
We consider the natural SDP relaxation for the
prefix discrepancy problem and show that the SDP discrepancy for every prefix is at most $1$.

The SDP discrepancy is a relaxation of discrepancy where
instead of signs $\varepsilon_j \in \{-1,1\}$ we allow vectors $w_j \in \R^d$ in some arbitrary dimension $d$ (as we have $n$ vectors $d = n$ always suffices) satisfying $\|w_j\|_2=1$. We say that the SDP prefix discrepancy is at most $C$ if there exist $w_1,\ldots,w_n$ such that
for every row $i = 1,2,\dotsc,m$ and for every $k=1,2, \dotsc,n$ it holds that
\begin{equation*}
    \Big\|\sum_{j=1}^k v^{(j)}_i w_j \Big\|_2 \leq  C \ .
\end{equation*}

For bounding the SDP prefix discrepancy
the idea is to use Banaszczyk's theorem for prefix discrepancy on a modified collection of vectors and a different convex body.
This idea was first pointed out to us by Raghu Meka in the context of SDP discrepancy for the Koml\'os problem (without prefixes). We adapt this idea to the prefix setting.

Let us first state Banaszczyk's theorem in the general form.
Recall, the Gaussian measure $\gamma_m(K)$ of a convex body $K\subseteq \mathbb R^m$
is the probability that a point $g\in\mathbb R^m$,
where each component is selected independently at random from the Gaussian
distribution $N(0, 1)$ with mean $0$~and variance~$1$,
is in $K$. We also say that $K$ is $0$-symmetric if $v\in K$ implies that $-v \in K$.
\begin{theorem}[\cite{DBLP:journals/rsa/Banaszczyk12}]
  Let $K\subseteq \mathbb R^m$ be $0$-symmetric and convex
  with $\gamma_m(K) \ge 1 - 1/2n$.
  Then for any $v^{(1)},v^{(2)},\dotsc,v^{(n)}$ there are signs
  $\epsilon_i\in \{-1, 1\}$ with
  \begin{equation*}
      \epsilon_1 v^{(1)} + \epsilon_2 v^{(2)} + \cdots + \epsilon_k v^{(k)} \in K
  \end{equation*}
  for all $k=1,2\dotsc,n$.
\end{theorem}
In the application to flow time scheduling (i.e., in the prefix Beck-Fiala and more generally the prefix Koml\'os setting), this theorem is used with a sufficiently
large hypercube.
Here we will define a more involved convex body $K$.

Let us fix an arbitrary small $\delta>0$. We will show that the SDP discrepancy is at most $1$.  
Let $r$ be an integer parameter that we will specify later. 
Let us replace each coordinate $i = 1,2,\dotsc,m$ by a {\em block} $B_i$ of $r$ coordinates, so that there are $rm$ coordinates in total. For $\ell = 1,2,\dotsc,r$ we use $(i,\ell)$ for the index of the $\ell$th coordinate in block $B_i$.
For each original vector $v^{(j)} \in \R^m$, we create $r$ vectors $v^{(j,\ell)} \in \R^{rm}$ with $\ell = 1,2,\dotsc,r$ as 
\begin{equation*}
    v^{(j,\ell)}_{i,\ell} = v^{(j)}_i \text{ for } i \in \{1,2,\dotsc, m\} \quad \text{ and } \quad v^{(j,\ell)}_{i,\ell'}=0 \text{ for } \ell' \neq \ell \text{ and } i \in \{1,2,\dotsc, m\} \ .
\end{equation*}
For a more visual definition, consider the matrix whose columns are $v^{(1)},v^{(2)},\dotsc,v^{(n)}$.
Let us replace each entry $v^{(j)}_i$ in this matrix by a $r \times r$ diagonal matrix with
$r$ copies of $v^{(j)}_i$ on its diagonal. The columns of the resulting matrix are the vectors
$v^{(j, \ell)}$.
Notice that the non-zero values in $v^{(j, \ell)}$ are the same as in $v^{(j)}$.
Hence, $\|v^{(j,\ell)}\|_2 = \|v^{(j)}\|_2\leq 1$.

For each $i = 1,2,\dotsc,m$ let $K_i$ be the convex body $K_i=\{ x\in \R^{rm} :  \sum_{\ell=1}^r x_{i,\ell}^2 \leq (1+\delta)^2 r\}$.
Let $K = \cap_{i=1}^m K_i$, that is,~$K$ is set of points $x \in \R^{rm}$ such that the sum of squares of the coordinates in each block $B_i$ is at most $(1+\delta)^2 r$.
\begin{lemma}
For large enough $r = O(\delta^{-2} \log mn)$,
the Gaussian measure $\gamma_{rm}(K) \geq 1-1/(2nr)$.
\end{lemma}
\begin{proof}
As $K = \cap_{i=1}^m K_i$ and $\gamma_{rm} (\mathbb R^{rm} \setminus K) = 1 - \gamma_{rm}(K)$, by the union bound we have \[1-\gamma_{rm}(K) \leq \sum_{i=1}^m (1-\gamma_{rm}(K_i))\] and thus it suffices to show that $\gamma_{rm}(K_i) \geq 1-1/(2nrm)$.

Let $g_1,\ldots,g_r$ be iid $N(0,1)$.
Then as $g_\ell^2$ has sub-exponential tails and $\E[g_\ell^2]=1,\E[g_\ell^4]=3$, by standard tail bounds for sum of independent sub-exponential random variables (see e.g.,~\cite{vershynin_2018}, Corollary 2.8.3)
there is a universal constant $c$ such that for any $t\geq 0$,
\begin{equation}
    \label{eq:conc-subexp}
 \Pr\Big[\sum_{\ell=1}^r g_\ell^2 - r \geq t r \Big]  \leq  \exp(- cr \min (t,t^2)) \ .
\end{equation} 
As the Gaussian measure of $K_i$ is exactly the probability that $\sum_{\ell=1}^r g_\ell^2 \leq (1+\delta)^2 r$,  
setting $t=2\delta + \delta^2 $ in \eqref{eq:conc-subexp} and choosing $r$ large enough so that $\exp(-cr \delta^2) \leq 1/(2nrm)$ gives that $\gamma_{rm}(K_i) \geq 1-1/(2nrm)$ and hence implies the claimed result.
\end{proof}
Let $\varepsilon_{j,\ell} \in \{-1,1\}$ be the signs obtained by applying Banaszczyk's prefix discrepancy theorem to the vectors
$v^{(1,1)},\ldots,v^{(1,r)}, v^{(2,1)},\ldots,v^{(2,r)},\ldots, v^{(n,1)},\ldots,v^{(n,r)}$ and the convex body $K$.
For $j = 1,2,\dotsc,n$ let us define the vector $w_j \in \R^r$ with entries
\begin{equation*}
    (w_j)_\ell = r^{-1/2} \epsilon_{j,\ell}
\end{equation*}
for $\ell = 1,2,\dotsc,r$. Clearly $\|w_j\|_2=1$ as each coordinate is $\pm r^{-1/2}$.
We now show that this is a good SDP solution.
\begin{lemma}
The vectors $w_1,w_2,\dotsc,w_n$ are a SDP solution with prefix discrepancy at most $C=1+\delta$.
\end{lemma} 
\begin{proof}
For $j=1,2,\dotsc,n$ let $u^{(j)} = \sum_{\ell = 1}^r \varepsilon_{j,\ell} v^{(j,\ell)} \in \R^{rm}$. By Banaszczyk's theorem  we have that $\sum_{j=1}^k u^{(j)} \in K$ for each $k = 1,2,\dotsc,n$.

Further, by the construction of the vectors $v^{(j,\ell)}$ we have the following property: for every $i=1,2,\dotsc,m$ and $j=1,2,\dotsc,n$, the vector $u^{(j)}$ restricted to the coordinates in block $B_i$ is exactly $r^{1/2} v^{(j)}_i w_j$. This follow as for each $\ell=1,2,\dotsc,r$ we have $u^{(j)}_{i,\ell}=  \eps_{j,\ell} v^{(j,\ell)}_{i,\ell} = v^{(j)}_i \eps_{j,\ell} $.

As $\sum_{j=1}^k u^{(j)} \in K$ for each $k=1,2,\dotsc,n$, and $K = \cap_{i = 1}^m K_i$, restricting to the coordinates in block $B_i$ gives that  $ \sum_{j=1}^k r^{1/2} v^{(j)}_i w_j$ has squared $\ell_2$-norm at most $(1+\delta)^2 r$ and hence that $\|  \sum_{j=1}^k v^{(j)}_i w_j \|_2 \leq 1+\delta$.
\end{proof}

\bibliographystyle{alpha}
\bibliography{references}

\newcommand{\etalchar}[1]{$^{#1}$}
\begin{thebibliography}{BMMP12}

\bibitem[AALR99]{DBLP:conf/stoc/AwerbuchALR99}
Baruch Awerbuch, Yossi Azar, Stefano Leonardi, and Oded Regev.
\newblock Minimizing the flow time without migration.
\newblock In {\em Proceedings of STOC}, pages 198--205, 1999.

\bibitem[ALS21]{ALS-STOC21}
Ryan Alweiss, Yang~P. Liu, and Mehtaab Sawhney.
\newblock Discrepancy minimization via a self-balancing walk.
\newblock In {\em Proceedings of STOC}, 2021.

\bibitem[Ban10]{DBLP:conf/focs/Bansal10}
Nikhil Bansal.
\newblock Constructive algorithms for discrepancy minimization.
\newblock In {\em Proceedings of FOCS}, pages 3--10, 2010.

\bibitem[Ban12]{DBLP:journals/rsa/Banaszczyk12}
Wojciech Banaszczyk.
\newblock On series of signed vectors and their rearrangements.
\newblock {\em Random Struct. Algorithms}, 40(3):301--316, 2012.

\bibitem[B{\'{a}}r81]{Barany81}
Imre B{\'{a}}r{\'{a}}ny.
\newblock A vector-sum theorem and its application to improving flow shop
  guarantees.
\newblock {\em Math. Oper. Res.}, 6(3):445--452, 1981.

\bibitem[BC16]{DBLP:journals/toc/BansalC16}
Nikhil Bansal and Bouke Cloostermans.
\newblock Minimizing maximum flow-time on related machines.
\newblock {\em Theory Comput.}, 12(1):1--14, 2016.

\bibitem[BCKL14]{DBLP:conf/soda/BansalCKL14}
Nikhil Bansal, Moses Charikar, Ravishankar Krishnaswamy, and Shi Li.
\newblock Better algorithms and hardness for broadcast scheduling via a
  discrepancy approach.
\newblock In {\em Proceedings of SODA}, pages 55--71, 2014.

\bibitem[BCM98]{DBLP:conf/soda/BenderCM98}
Michael~A. Bender, Soumen Chakrabarti, and S.~Muthukrishnan.
\newblock Flow and stretch metrics for scheduling continuous job streams.
\newblock In {\em Proceedings of SODA}, pages 270--279, 1998.

\bibitem[BDG19]{DBLP:journals/siamcomp/BansalDG19}
Nikhil Bansal, Daniel Dadush, and Shashwat Garg.
\newblock An algorithm for koml{\'{o}}s conjecture matching banaszczyk's bound.
\newblock {\em {SIAM} J. Comput.}, 48(2):534--553, 2019.

\bibitem[BDGL19]{DBLP:journals/toc/BansalDGL19}
Nikhil Bansal, Daniel Dadush, Shashwat Garg, and Shachar Lovett.
\newblock The gram-schmidt walk: {A} cure for the banaszczyk blues.
\newblock {\em Theory Comput.}, 15:1--27, 2019.

\bibitem[BDJR22]{BerndtDJRAlenex}
Sebastian Berndt, Max~A. Deppert, Klaus Jansen, and Lars Rohwedder.
\newblock Load balancing: The long road from theory to practice.
\newblock In {\em Proceedings of ALENEX}. {SIAM}, 2022.
\newblock to appear.

\bibitem[BF81]{BeckFiala-DAM81}
J{\'o}zsef Beck and Tibor Fiala.
\newblock {``Integer-making'' theorems}.
\newblock {\em Discrete Appl. Math.}, 3(1):1--8, 1981.

\bibitem[BG81]{BaranyGrinberg81}
Imre B{\'a}r{\'a}ny and Victor~S Grinberg.
\newblock On some combinatorial questions in finite-dimensional spaces.
\newblock {\em Linear Algebra and its Applications}, 41:1--9, 1981.

\bibitem[BG17]{DBLP:conf/stoc/BansalG17}
Nikhil Bansal and Shashwat Garg.
\newblock Algorithmic discrepancy beyond partial coloring.
\newblock In {\em Proceedings of STOC}, pages 914--926, 2017.

\bibitem[BGK18]{DBLP:conf/focs/Batra0K18}
Jatin Batra, Naveen Garg, and Amit Kumar.
\newblock Constant factor approximation algorithm for weighted flow time on a
  single machine in pseudo-polynomial time.
\newblock In {\em Proceedings of FOCS}, pages 778--789, 2018.

\bibitem[BJM{\etalchar{+}}21]{BJMSS-SODA21}
Nikhil Bansal, Haotian Jiang, Raghu Meka, Sahil Singla, and Makrand Sinha.
\newblock Online discrepancy minimization for stochastic arrivals.
\newblock In {\em Proceedings of SODA}, pages 2842--2861, 2021.

\bibitem[BJM{\etalchar{+}}22]{bjmss22-itcs}
Nikhil Bansal, Haotian Jiang, Raghu Meka, Sahil Singla, and Makrand Sinha.
\newblock Prefix discrepancy, smoothed analysis, and combinatorial vector
  balancing.
\newblock In {\em Proceedings of ITCS}, 2022.
\newblock no appear.

\bibitem[BJSS20]{BJSS20}
Nikhil Bansal, Haotian Jiang, Sahil Singla, and Makrand Sinha.
\newblock Online vector balancing and geometric discrepancy.
\newblock In {\em Proceedings of STOC}, pages 1139--1152, 2020.

\bibitem[BK15]{DBLP:conf/stoc/BansalK15}
Nikhil Bansal and Janardhan Kulkarni.
\newblock Minimizing flow-time on unrelated machines.
\newblock In {\em Proceedings of STOC}, pages 851--860, 2015.

\bibitem[BMMP12]{BMMP12}
Kevin Buchin, Ji\v{r}{\'{\i}} Matou\v{s}ek, Robin~A. Moser, and
  D{\"{o}}m{\"{o}}t{\"{o}}r P{\'{a}}lv{\"{o}}lgyi.
\newblock Vectors in a box.
\newblock {\em Math. Program.}, 135(1-2):323--335, 2012.

\bibitem[BP14]{DBLP:journals/siamcomp/BansalP14}
Nikhil Bansal and Kirk Pruhs.
\newblock The geometry of scheduling.
\newblock {\em {SIAM} J. Comput.}, 43(5):1684--1698, 2014.

\bibitem[Cho94]{Chobanyan94}
Sergej Chobanyan.
\newblock {Convergence as of rearranged random series in Banach space and
  associated inequalities}.
\newblock In {\em Probability in Banach Spaces, 9}, pages 3--29. Springer,
  1994.

\bibitem[DGG99]{DinitzGG99}
Yefim Dinitz, Naveen Garg, and Michel~X. Goemans.
\newblock On the single-source unsplittable flow problem.
\newblock {\em Comb.}, 19(1):17--41, 1999.

\bibitem[EW18]{EW18}
Friedrich Eisenbrand and Robert Weismantel.
\newblock Proximity results and faster algorithms for integer programming using
  the steinitz lemma.
\newblock In {\em Proceedings of SODA}, pages 808--816, 2018.

\bibitem[FKL19]{DBLP:conf/soda/FeigeKL19}
Uriel Feige, Janardhan Kulkarni, and Shi Li.
\newblock A polynomial time constant approximation for minimizing total
  weighted flow-time.
\newblock In {\em Proceedings of SODA}, pages 1585--1595, 2019.

\bibitem[GK06]{DBLP:conf/icalp/GargK06}
Naveen Garg and Amit Kumar.
\newblock Better algorithms for minimizing average flow-time on related
  machines.
\newblock In {\em Proceedings of ICALP}, volume 4051, pages 181--190, 2006.

\bibitem[GK07]{DBLP:conf/focs/GargK07}
Naveen Garg and Amit Kumar.
\newblock Minimizing average flow-time : Upper and lower bounds.
\newblock In {\em Proceedings of FOCS}, pages 603--613, 2007.

\bibitem[GKM08]{DBLP:conf/isaac/GargKM08}
Naveen Garg, Amit Kumar, and V.~N. Muralidhara.
\newblock Minimizing total flow-time: The unrelated case.
\newblock In {\em Proceedings of ISAAC}, volume 5369, pages 424--435, 2008.

\bibitem[HR17]{DBLP:conf/soda/HobergR17}
Rebecca Hoberg and Thomas Rothvoss.
\newblock A logarithmic additive integrality gap for bin packing.
\newblock In {\em Proceedings of SODA}, pages 2616--2625, 2017.

\bibitem[JR19]{JR-ITCS19}
Klaus Jansen and Lars Rohwedder.
\newblock On integer programming and convolution.
\newblock In {\em Proceedings of ITCS}, pages 43:1--43:17, 2019.

\bibitem[LM15]{DBLP:journals/siamcomp/LovettM15}
Shachar Lovett and Raghu Meka.
\newblock Constructive discrepancy minimization by walking on the edges.
\newblock {\em {SIAM} J. Comput.}, 44(5):1573--1582, 2015.

\bibitem[LR07]{DBLP:journals/jcss/LeonardiR07}
Stefano Leonardi and Danny Raz.
\newblock Approximating total flow time on parallel machines.
\newblock {\em J. Comput. Syst. Sci.}, 73(6):875--891, 2007.

\bibitem[LST90]{DBLP:journals/mp/LenstraST90}
Jan~Karel Lenstra, David~B. Shmoys, and {\'{E}}va Tardos.
\newblock Approximation algorithms for scheduling unrelated parallel machines.
\newblock {\em Math. Program.}, 46:259--271, 1990.

\bibitem[LSV86]{DBLP:journals/ejc/LovaszSV86}
L{\'{a}}szl{\'{o}} Lov{\'{a}}sz, Joel Spencer, and Katalin Vesztergombi.
\newblock Discrepancy of set-systems and matrices.
\newblock {\em Eur. J. Comb.}, 7(2):151--160, 1986.

\bibitem[Mas04]{DBLP:journals/ijfcs/Mastrolilli04}
Monaldo Mastrolilli.
\newblock Scheduling to minimize max flow time: Off-line and on-line
  algorithms.
\newblock {\em Int. J. Found. Comput. Sci.}, 15(2):385--401, 2004.

\bibitem[Mor21]{MorellS21}
Martin Morell, Sarah~Skutella.
\newblock Single source unsplittable flows with arc-wise lower and upper
  bounds.
\newblock {\em Mathematical Programming}, 2021.

\bibitem[Nik17]{Nikolov-Mathematika19}
Aleksandar Nikolov.
\newblock Tighter bounds for the discrepancy of boxes and polytopes.
\newblock {\em Mathematika}, 63(3):1091--1113, 2017.

\bibitem[Rot13]{DBLP:conf/focs/Rothvoss13}
Thomas Rothvo{\ss}.
\newblock Approximating bin packing within {O}(log {OPT} * log log {OPT)} bins.
\newblock In {\em Proceedings of FOCS}, pages 20--29, 2013.

\bibitem[RW21]{DBLP:conf/stoc/RohwedderW21}
Lars Rohwedder and Andreas Wiese.
\newblock A {(2} + \emph{{\(\epsilon\)}})-approximation algorithm for
  preemptive weighted flow time on a single machine.
\newblock In {\em Proceedings of STOC}, pages 1042--1055, 2021.

\bibitem[Sev94]{Sevastjanov94}
Sergey Sevast'janov.
\newblock On some geometric methods in scheduling theory: {A} survey.
\newblock {\em Discret. Appl. Math.}, 55(1):59--82, 1994.

\bibitem[Sin21]{sahil-personal}
Sahil Singla, 2021.
\newblock Personal communication.

\bibitem[Spe77]{Spencer77}
Joel Spencer.
\newblock Balancing games.
\newblock {\em J. Comb. Theory, Ser. B}, 23(1):68--74, 1977.

\bibitem[Spe86]{Spencer86Prefix}
Joel Spencer.
\newblock Balancing vectors in the max norm.
\newblock {\em Combinatorica}, 6(1):55--65, 1986.

\bibitem[Spe94]{spencer1994ten}
Joel Spencer.
\newblock {\em Ten lectures on the probabilistic method}.
\newblock SIAM, 1994.

\bibitem[Ste16]{Steinitz-16}
Ernst Steinitz.
\newblock Bedingt konvergente reihen und konvexe systeme.
\newblock {\em Journal f{\"u}r die reine und angewandte Mathematik}, 146:1--52,
  1916.

\bibitem[Ver18]{vershynin_2018}
Roman Vershynin.
\newblock {\em High-Dimensional Probability: An Introduction with Applications
  in Data Science}.
\newblock Cambridge University Press, 2018.

\end{thebibliography}



\end{document}